\documentclass[acmsmall]{acmart}

\AtBeginDocument{%
  \providecommand\BibTeX{{%
    \normalfont B\kern-0.5em{\scshape i\kern-0.25em b}\kern-0.8em\TeX}}}

\setcopyright{rightsretained}
\copyrightyear{20xx}
\acmYear{20yy}

\usepackage{hyperref}
\usepackage{xfrac}
\usepackage{prettyref}

\newrefformat{sec}{Section\,\ref{#1}}
\newrefformat{app}{Appendix\,\ref{#1}}
\newrefformat{def}{Def.\,\ref{#1}}
\newrefformat{thm}{Theorem\,\ref{#1}}
\newrefformat{prop}{Proposition\,\ref{#1}}
\newrefformat{lem}{Lemma\,\ref{#1}}
\newrefformat{cor}{Corollary\,\ref{#1}}
\newrefformat{rem}{Remark\,\ref{#1}}
\newrefformat{ex}{Example\,\ref{#1}}
\newrefformat{tab}{Table\,\ref{#1}}
\newrefformat{fig}{Fig.\,\ref{#1}}
\newrefformat{case}{case\,\ref{#1}}
\newrefformat{foot}{Footnote\,\ref{#1}}

\usepackage{amsthm}

\usepackage{amsmath}
\usepackage{amssymb}
\usepackage{natbib}
\usepackage{mathtools}
\usepackage{tensor}
\usepackage{bm}
\usepackage{tikz}
\usepackage{tikz-cd}
\usepackage{graphicx}
\usepackage{stmaryrd}
\usepackage{caption}
\usepackage{subcaption}
\usepackage{algorithm}
\usepackage{algpseudocode}

\definecolor{semblue}{rgb}{0,0,0.7}
\definecolor{vgreen}{rgb}{.1,.5,0}%
\definecolor{vdarkgreen}{rgb}{.06,.3,0}%
\definecolor{vred}{rgb}{.7,0,0}%
\definecolor{vblue}{rgb}{.1,.15,.62}%
\definecolor{vgray}{rgb}{.35,.35,.35}
\definecolor{darkishgray}{rgb}{.35,.35,.35}
\definecolor{vvblue}{rgb}{.14,.21,.868}%

\usepackage{xcolor}
\usepackage{float}
\usepackage[prefixflatinterpret,prefixflatbracketmodalinterpret,prefixflatsetinterpret,sidenotecalculus,longseqcontext]{logic}
\usepackage[prefixflatinterpret,prefixflatbracketmodalinterpret,differentialdL,simplenames]{dL}

\graphicspath{ {./Images/} }

\usepackage{xspace}
\usepackage{xifthen}

\newtheorem{theorem}{Theorem}[section]

\newtheorem{proposition}[theorem]{Proposition}
\newtheorem{corollary}[theorem]{Corollary}
\newtheorem{lemma}[theorem]{Lemma}
\newtheorem{definition}[theorem]{Definition}

\theoremstyle{definition}
\newtheorem{example}[theorem]{Example}
\newtheorem{remark}[theorem]{Remark}

\newcommand{\R}{\mathbb{R}}

\newcommand{\N}{\mathbb{N}}
\newcommand{\Q}{\mathbb{Q}}
\newcommand{\V}{\mathbb{V}}

\renewcommand{\S}{\mathbb{S}}

\renewcommand{\phi}{\varphi}

\newcommand{\eps}{\varepsilon}
\newcommand{\eval}[1]{\left\llbracket #1 \right\rrbracket}

\newcommand{\norm}[1]{\left\lVert#1\right\rVert}
\newcommand{\bebecomes}{\mathrel{::=}}
\newcommand{\alternative}{~|~}

\newcommand{\rels}{\{\geq, >\}}

\newcommand{\graph}{\text{graph}}
\newcommand{\safe}[4][]{
    \ifthenelse{\isempty{#1}}
    {\text{SAFE}(#2, #3, #4)}
    {\text{SAFE}_{#1}(#2, #3, #4)}
}

\newcommand{\unsafe}[4][]{
    \ifthenelse{\isempty{#1}}
    {\text{UNSAFE}(#2, #3, #4)}
    {\text{UNSAFE}_{#1}(#2, #3, #4)}
}
\newcommand{\lp}{\text{LP}}

\newcommand{\folr}{\text{FOL}_{\R}}

\newcommand{\fvarA}{\phi}
\newcommand{\fvarB}{\psi}

\newcommand{\leftrule}{L}
\newcommand{\rightrule}{R}

\newsavebox{\Rval}%
\sbox{\Rval}{$\scriptstyle\mathbb{R}$}
\newsavebox{\backiterateb}%
\sbox{\backiterateb}{$\scriptstyle\overleftarrow{\dibox{{}^*}}$}

\newcommand{\rfvar}{P}

\newcommand{\D}[1]{#1'}
\renewcommand{\safe}{\text{SAFE}}
\newcommand{\safer}{\text{SAFE}_{\text{R}}}

\begin{document}

\renewcommand{\linferPremissSeparation}{\hspace{0.8cm}}

\title{Differential Equation Inductive Robustness Axiomatization}
\author{Andr\'e Platzer}
\author{Long Qian}
\email{platzer@kit.edu}
\email{longq@andrew.cmu.edu}
\orcid{0000-0001-7238-5710}
\orcid{0000-0003-1567-3948}
\affiliation{%
  \institution{Karlsruhe Institute of Technology}
  \city{Karlsruhe}
  \country{Germany}
  }
\affiliation{%
\institution{Carnegie Mellon University}
\city{Pittsburgh}
\country{USA}
}

\begin{abstract}
    This article establishes the completeness of an axiomatization for the robust safety of dynamical systems with polynomial differential equations on bounded time horizons. Safety properties of robust systems are uniformly reduced to a sound axiomatization of polynomial invariants, resulting in reliable logical proofs of correctness. Approximate decidability results are also established: there is a computable algorithm such that, given any perturbation parameter $\delta$, it either produces a symbolic proof of robust safety (hence correctly decides the dynamical system to be robustly safe), or correctly decides that the system is not robustly safe under a perturbation of level $\delta$. In contrast to earlier works, this article crucially leverages results from subanalytic geometry to retain a level of exactness, thereby establishing positive results of provability/decidability allowing for arbitrary bounded (semialgebraic) initial/post conditions even \emph{without positive separation} at their (topological) boundaries. This enables the generation of proofs of \emph{inductive safety} beyond finite time horizons for general \emph{hybrid dynamical systems}. 
\end{abstract}



\maketitle
\section{Introduction}

The verification of safety properties for ordinary differential equations (ODEs) is of fundamental importance and a vital part of verifying cyber-physical systems (CPS) \cite{Alur_2015,DBLP:books/sp/Platzer18}. Such analyses of ODEs can be carried out quantitatively by computing approximations to their solutions from specific initial states \cite{DBLP:conf/cav/ChenAS13, DBLP:conf/cav/FrehseGDCRLRGDM11}, or qualitatively by directly reasoning with their vector fields. This article contributes to the qualitative study of ODEs via their \emph{logical axiomatizations}, where properties of ODEs are deductively proven from a small set of sound core axioms. Such axiomatic treatments are desirable as they yield trustworthy symbolic proofs of correctness, providing reliable guarantees of safety that can be independently verified. 

This article establishes the logical completeness of (bounded, semialgebraic) \emph{robust safety} for ODEs with polynomial vector fields using \emph{differential dynamic logic} (\dL) \cite{DBLP:journals/jar/Platzer08, DBLP:conf/lics/Platzer12b}, providing complete logical foundations for the safety verification of \emph{robust} ODEs - ODEs whose safety does not depend on arbitrarily small perturbations. By simultaneously leveraging both exact symbolics and numerical approximations, completeness is even attained for the subtle case where the set of initial conditions is \emph{not} at a positive distance away from the unsafe states, thus safety cannot be verified by numerical approximations alone as no margin of error is available. Note that similar results for exact completeness without assuming robustness are fundamentally challenging, as they would imply the decidability of the bounded Skolem-Pisot problem \cite{DBLP:journals/tcs/BellDJB10}.

This article is concerned with the fragment of $\dL$ comprising the following formulas expressing safety properties of polynomial differential equations \(x' = p(x)\) for formulas $I, S \in \folr$ of real arithmetic as initial and safety conditions and a time bound $T \in \Q^+$:
\[\safe(I, S) ~\equiv~ I \land t = 0 \rightarrow \dbox{x' = p(x), t' = 1 \& t \leq T}{S}\]
Mathematically, the validity of the formula $\safe(I, S)$ is equivalent to the statement that for every initial condition $x_0 \in \eval{I}$ (where $\eval{I} = \{x \in \R^n \mid \R \models I(x)\}$ is the subset of $\R^n$ defined by the formula $I(x)$), $x_0$ \emph{always remains} in the set $S$ under the flow $x'(t) = p(x(t))$ for $0 \leq t \leq T$, where $p(x) \in \Q[x]$ is a rational polynomial in $x$. The modality $\dbox{x' = p(x), t' = 1 \& t \leq T}{S}$ is a real-time extension of the classical box modality $\Box S$ in modal logic, indicating that the safety property $S$ is always satisfied after all solutions of the differential equation \(x'=p(x), t'=1\) that always satisfy the formula $t \leq T$ restricting the time horizon to $[0, T]$.

This article is concerned with \emph{robust safety}, the property that $\safe(I, S)$ remains true under arbitrarily small perturbations. Such robust safety properties are here defined as follows, where $S^o$ denotes a $\folr$ formula defining the topological interior of $\eval{S}$ and similarly $\overline{I}$ denotes a formula defining the topological closure of $\eval{I}$. 
\[\safer(I, S) ~\equiv~ (I \rightarrow S) \land (\overline{I} \land t = 0 \rightarrow \dbox{x' = p(x), t' = 1 \& t \leq T}{(t > 0 \rightarrow S^o)})\]
That is, the set of initial conditions $I$ is robustly safe if it is safe at time $t = 0$ and the states reachable along the differential equation from the closure $\overline{I}$ are contained in the interior of the safety set $S$ for all \emph{positive times}. This robustness property is also referred to as \emph{topological robustness}. The main result of this article is that $\dL$ is \emph{complete} for $\safer(I, S)$ (for bounded $\eval{I}, \eval{S}$) and there exists a computable correspondence between valid \dL formulas of the form $\safer(I, S)$ and their proofs - given a robust safety formula $\safer(I, S)$ that is valid, one can computably find a deductive proof of it in \dL. 

\paragraph{Comparison.}
Although the topological robustness property $\safer(I, S)$ might appear similar to the stronger condition $\safe(\overline{I}, S^o)$, which also appears to be a natural definition for robust safety, the crucial difference is that at time $t = 0$ the condition $\safer(I, S)$ reduces down to the safety of initial conditions $I \subseteq S$ only, and \emph{not} the stronger safety requirement $\overline{I} \subseteq S^o$. Importantly, this allows for $\safer(I, S)$ to hold even when\footnote{This article uses $d(A, B) \coloneqq \inf_{a \in A}\inf_{b \in B} d(a, b)$} $d(I, S^c) = 0$ and $I$ is \emph{not positively separated} from the complement $S^c$ so that no extra margin of error is available. This seemingly minor change is the difference between allowing for \emph{inductive} safety or not. Intuitively, if robust safety was defined via $\safe(\overline{I}, S^o)$, then no non-trivial $I$ can satisfy $\safe(\overline{I}, I^o)$ as safety is always violated at $t = 0$ because $\overline{I} \nsubseteq I^o$, and therefore no robust inductive invariants exist.

In contrast, such obstructions are not present in the case of topological robustness, precisely because initially at time $t = 0$ only $I \subseteq S$ is required. Fundamentally, $\safer(I, S)$ allows for $I$ to \emph{not be} positively separated from the unsafe sets, whereas $\safe(\overline{I}, S^o)$ \emph{requires} a positive separation and thereby expects additional robustness beyond what its initial condition itself provides, making $\safer(I, S)$ an \emph{inductive} notion of robustness in contrast to the inherently non-inductive $\safe(\overline{I}, S^o)$. 

\paragraph{Inductiveness.}
The inductiveness of topological robustness is particularly pronounced in the safety verification of CPS or \emph{hybrid} dynamical systems, where the underlying continuous dynamics $x' = p(x)$ is generalized to $x' = p(x, u), u' = 0$ with $u$ representing the \emph{control variables} from discrete assignments. Such continuous dynamics are executed over (infinitely) many \emph{discrete} steps, after each execution, the value of the control variables $u$ might change due to actions of discrete assignments. That is, the continuous variables $x$ evolve continuously, whereas the control variables $u$ evolve via discrete transitions, making the overall dynamical system a hybrid dynamical system \cite{DBLP:conf/hybrid/NerodeK92a}. The safety verification of such hybrid dynamical systems amounts to verifying that the desired safety requirement holds after \emph{all possible executions}. The following toy example illustrates one such instance (Figure \ref{fig: cart}).

\begin{figure}
    
    \centering
    \includegraphics[width=\linewidth]{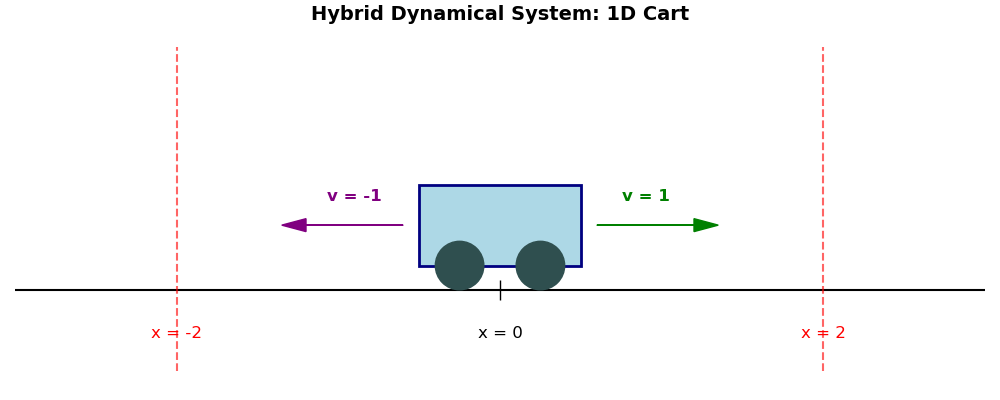}
    \captionof{figure}{1D Cart}
    \label{fig: cart}
\end{figure}

\begin{algorithm}
\caption{Example of Hybrid Dynamical System}
\label{eg: hybrid cart}
\begin{algorithmic}[1]
\State \textbf{Safety requirement:} $-2 \leq x \leq 2$
\State \textbf{Initial conditions:} $x = 0 \land v = 0$
\While{\textbf{True}}
    \If{$x \geq 1$}
        \State $v \gets -1$
    \ElsIf{$x \leq -1$}
        \State $v \gets 1$
    \Else
        \State $v \in [-1, 1]$ (chosen arbitrarily)
    \EndIf
    \State $t \gets 0$
    \State $\left\{x' = v, v' = 0, t' = 1 \& t \leq \frac{1}{2}\right\}$ \Comment{Continuous Dynamics}
\EndWhile
\end{algorithmic}
\end{algorithm}
Variable $v$ evolves discretely following the logical conditionals and assignments, and the variable $x$ evolves continuously along the ODE $x' = v$. The safety requirement of the overall system is $-2 \leq x \leq 2$ with initial condition $x = 0$. The standard approach to safety verification is to use \emph{loop invariants}, constructing a formula $j \in \folr$ that remains invariant under each iteration of the dynamical system. To illustrate the inductiveness of topological robustness, consider the simple loop invariant $j \equiv -2 \leq x \leq 2 \land -2 \leq v \leq 2$, in which case the invariance conditions reduce down to the following \dL formulas:
\begin{enumerate}
    \item Initial condition: $x = 0 \land v = 0 \rightarrow j$
    \item Terminal condition: $j \rightarrow -2 \leq x \leq 2$
    \item Inductive condition:
    \[j \land t = 0 \land -1 \leq v \leq 1 \land \left(x \geq 1 \rightarrow v = -1\right) \land \left(x \leq -1 \rightarrow v = 1\right) \rightarrow \dbox{x' = v, v' = 0, t' = 1 \& t \leq 1/2}{j}\]
\end{enumerate}
The first two conditions are standard to ensure that $j$ is implied by the initial conditions while also being strong enough to ensure the safety condition, the last condition ensures that $j$ is \emph{inductive} and the validity of $j$ is preserved after one execution of the loop body. It can be observed that the constructed invariant $j \equiv -2 \leq x \leq 2 \land -2 \leq v \leq 2$ trivially satisfies the first two proof obligations. Consider condition (3), which is a safety problem of the form $\safe(I, S)$ where the bounded sets $I, S$ are defined via
\begin{align*}
    I &\equiv j \land -1 \leq v \leq 1 \land \left(x \geq 1 \rightarrow v = -1\right) \land \left(x \leq -1 \rightarrow v = 1\right)\\
    S &\equiv j
\end{align*}
Importantly, note that $(x, v) = (-2, 1) \in \eval{I} \cap \eval{S}$ is a point on the \emph{boundary} of $\eval{S}$, therefore $\safe(I, S)$ \emph{cannot} be proven by invoking any completeness result that does not distinguish $\safe(I, S)$ from $\safe(\overline{I}, S^o)$ because the first is valid and the second is invalid. Concretely, this implies that defining $\safe(I, S)$ to be ``robustly valid'' if $\safe(\overline{I}, S^o)$ is valid results in a notion of robustness that is not inductive. By contrast, it is not hard to see that $\safer(I, S)$ is indeed valid, and therefore provable by the completeness result established in this article. While the dynamics of $x' = v$ is simple and solvable in this case, the results of this article establish the general completeness of topological robustness, even for much more complicated non-linear dynamics. 

Consequently, it follows from the results of this article that $\dL$ is complete for robust loop invariants, given a loop invariant $j$ such that the corresponding safety problem arising from the inductive condition is robustly valid, the overall safety will be provable in $\dL$. Fundamentally, the (only) difference between $\safer(I, S)$ and $\safe(\overline{I}, S^o)$ at $t = 0$ of allowing there to be no positive separation is precisely what enables the former notion to be inductive and not the latter.

Similar to how the seemingly minor change at $t = 0$ results in $\safer(I, S)$ being substantially better-behaved and inductive, this change also results in a substantially more complicated proof of completeness compared to $\safe(\overline{I}, S^o)$. Indeed, the completeness of $\safe(\overline{I}, S^o)$ has been established in earlier work \cite{DBLP:journals/jacm/PlatzerQ25}, essentially by proving that \dL is ``complete for numerical approximations'' and utilizing the fact that a positive separation at all times (including $t = 0$) implies that sufficiently accurate numerical approximations to the flow of $x' = p(x)$ will witness the validity of $\safe(\overline{I}, S^o)$. However, such numerical approximations are established in the framework of computable analysis, where computations (of real-numbers) are inherently inexact. Fundamentally, when merely operating on the level of (real) computability one cannot (computably) distinguish a set $S \subseteq \R^n$ from its topological closure $\overline{S}$, and therefore cannot distinguish the validity of the formulas $\safer(I, S), \safer(\overline{I}, S^o)$, but this is precisely the challenge as it is possible for $\safer(I, S)$ to be valid and $\safer(\overline{I}, S^o)$ be invalid due to the different requirements at $t = 0$ ($\eval{I} \subseteq \eval{S}$ compared to $\overline{\eval{I}} \subseteq \eval{S}^o$). Thus, computability-theoretic methods via numerical approximations will not suffice to establish completeness of $\safer(I, S)$.

This article precisely combines the advantages of both computability-theoretic approaches to robust ODEs and axiomatic treatments of ODEs, such that a complete axiomatization for $\safer(I, S)$ is obtained (Figure \ref{fig: proof_strat}). Intuitively, completeness is obtained by suitably partitioning the time horizon of interest $[0, T]$ into two parts $[0, s], [s, T]$ for some $s \in \Q^+$ sufficiently small. On the small interval $[0, s]$, safety is established by leveraging the symbolic nature of polynomial ODEs with sufficiently accurate Taylor approximations and the $\folr$ definability of $I, S$ to prove that the trajectory is away from the unsafe set $S^c$. On the interval $[s, T]$, safety can then be established by computability-theoretic methods with sufficiently accurate numerical approximations. The primary difficulty is in establishing the existence of such a rational constant $s$, which uses results from subanalytic geometry (\L{}ojasiewicz's inequality for continuous subanalytic functions, Theorem \ref{thm: subanalytic inequality}) to derive suitable local progress bounds (Theorem \ref{thm: compact progression characterization}). This combination of deductive logic with computable analysis  provides the desired \emph{complete logical foundation for the safety verification of robust differential equations} justified by subanalytic geometry. 

\begin{figure}
    \centering
    \includegraphics[width=\linewidth]{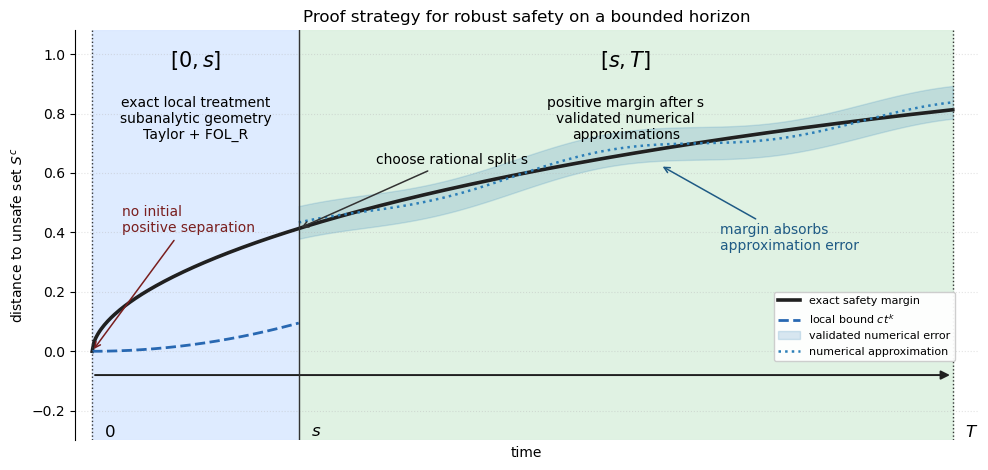}
    \captionof{figure}{Proof Strategy}
    \label{fig: proof_strat}
\end{figure}

\paragraph{Robust Relations.}
While the primary result of this article is to establish a complete axiomatization for robust safety of ODEs that is inductive, it also offers an interesting perspective on the different ways one can define what it means for a safety problem $\safe(I, S)$ to be ``robustly valid'' over a finite time horizon $[0, T]$:

\begin{itemize}
    \item \emph{state robustness:} $\safe(I, S)$ is robustly valid if $\safe(\overline{I}, S^o)$ is valid.
    \item  \emph{dynamics robustness:} $\safe(I, S)$ is robustly valid if safety holds by relaxing the dynamics $x' = p(x)$ to differential inclusions $\norm{x' - p(x)} < \eps$ for some $\eps \in \Q^+$ \cite{DBLP:journals/tac/Ratschan18}. 
    \item \emph{topological robustness:} $\safe(I, S)$ is robustly valid if $\safer(I, S)$ is valid. This robustness notion is presented in this article.  
\end{itemize}

\begin{table}[h]
    \centering
    \renewcommand{\arraystretch}{1.5}%
    \caption{Summary of robustness properties compared to the new notion of topological robustness $\safer(I, S)$}
    \label{tab: robust_comp}
    \begin{tabular}{|l|c|c|c|}
        \hline
        & \textbf{State Robustness} & \textbf{Dynamics Robustness} & \textbf{Topological Robustness} \\ \hline
        
        \textbf{Axiomatization} & $\checkmark$ & $\times$ & $\checkmark$ \\ \hline
        
        \textbf{Inductive}     & $\times$ & $\checkmark$ & $\checkmark$ \\ \hline
    \end{tabular}
\end{table}

State robustness also admits a complete axiomatization \cite{DBLP:journals/jacm/PlatzerQ25} but lacks the inductive property. By comparison, it is not hard to see that dynamics robustness can be made inductive, but is challenging to obtain a complete logical axiomatization due to the use of differential inclusions, which essentially quantifies over all $C^1$ solutions and does not admit adequate logical frameworks like $\dL$. From this perspective, the definition of topological robustness with $\safer(I, S)$ attains the best of both worlds - it is inductive while also exhibiting a complete logical axiomatization (cf. Table \ref{tab: robust_comp}). Furthermore, as the underlying time horizon $[0, T]$ is compact, these different notions of robustness can be directly compared, resulting in the following inclusions.
\[\text{State Robustness} \subsetneq \text{Dynamics Robustness} \subsetneq \text{Topological Robustness}\]

Topological robustness $\safer(I, S)$ is also the weakest and most general notion of robustness among the three. While the definition of topological robustness may at first appear less natural compared to the more familiar notions of state/dynamics robustness, the topological characterization of $\safer(I, S)$ is in fact a natural generalization of dynamics robustness, and can be viewed as the limit of a hierarchy of increasingly more general differential inclusions. Consider the modified notions of dynamics robustness $\safe_n(I, S)$ that hold if the differential inclusion $\norm{x' - p(x)} < \eps t^n$ on $(0, T]$ (where $t$ denotes the time variable) is safe for some $\eps \in \Q^+$, with $n \in \N$. Then $\safe_0(I, S)$ is exactly dynamics robustness, $\safe_1(I, S)$ corresponds to the differential inclusion $\norm{x' - p(x)} < \eps t$ and so on. It is not hard to see that for all $n$, $\safe_n(I, S)$ implies $\safe_{n + 1}(I, S)$ over bounded time horizons, and such implications can also be strict. For example, consider the simple dynamics
\[x' = v, v' = 1\]
where the safety requirement is $S(x, v) \equiv x \geq 0$ on the time interval $[0, 1]$ and the initial conditions are $I(x, v) \equiv x = 0 \land v = 0$. $\safe_0(I, S)$ does not hold since for all $\eps \in \Q^+$ the dynamics $x' = v - \eps, v' = 1$ will be unsafe. By contrast, it is not hard to see that $\safe_1(I, S)$ does hold with a perturbation of $\frac{t}{4}$. In light of this example, it is natural to consider the ``most general'' notion of dynamics robustness, denoted $\safe_\infty(I, S)$, by taking the infinite disjunction of $\safe_n(I, S)$ across all $n$. That is, $\safe_\infty(I, S)$ is valid if and only if there exists some $n \in \N$ such that $\safe_n(I, S)$ is valid.
\[\safe_\infty(I, S) \equiv \bigvee_{n \in \N} \safe_n (I, S)\]
It turns out that this is exactly equivalent to $\safer(I, S)$ (Theorem \ref{thm: limiting characterization of robust safety}). Hence, not only is the definition of $\safer(I, S)$ natural, it can also be viewed as the robust safety of the limiting dynamics of a hierarchy of differential inclusions:
\[\text{State Robustness} \subsetneq \text{Dynamics Robustness} \equiv \safe_0 \subsetneq \safe_1 \subsetneq \dots \subsetneq \safe_\infty \equiv \safer\]

\paragraph{Contributions.}
This article also establishes approximate decidability results for topological robustness, as well as extensions of both completeness and approximate decidability to ODEs with \emph{domain constraints}, where the flow of $x' = p(x)$ is constrained to some semialgebraic set $Q$.
In summary, this article establishes the following main results: 

\begin{enumerate}
    \item {\bf{Completeness for robust safety}}: $\dL$ is complete for robust safety, for all bounded $I, S \in \folr$ and $T \in \Q^+$ validity and provability of $\safer(I, S)$ coincide.
    \[\models \safer(I, S) ~{\iff}~~~ \vdash \safer(I, S)\]
    Furthermore, this completeness is effective. There exists a direct computable algorithm that produces a proof of $\safer(I, S)$ in $\dL$ provided that $\safer(I, S)$ is valid. 

    \item {\bf{$\delta$-decidability for robust safety}}: There exists a computable algorithm such that, given a (computable) perturbation function $\delta \in C^0([0, T], \R^{\geq 0})$ and a robust safety problem $\safer(I, S)$, it correctly outputs one of the following
    \begin{itemize}
        \item $\safer(I, S)$ is valid and $I$ is robustly safe. 
        \item $\safer(I, S)$ is not valid under a perturbation of $\delta$ (not $\delta$-safe).
    \end{itemize}
    Furthermore, this characterization is exact, $\safer(I, S)$ is valid if and only if it is $\delta$-safe for some $\delta$. 

    \item {\bf{Extensions to domain constraints}}: This article also discusses extensions of the results above to the case of constrained evolution, where flows of $x' = p(x)$ are constrained to stop before they leave a given semialgebraic set $Q \in \folr$. Earlier completeness results \cite{DBLP:journals/jacm/PlatzerQ25, DBLP:conf/ijcar/AbouElWafaP26} concerning $\safe(I, S)$ for $I$ compact and $S$ open are also generalized to allow for topologically closed constraint sets that are not necessarily bounded. 
\end{enumerate}

\section{Related Work}

This article establishes a complete logical axiomatization for robust safety of (polynomial) ODEs, utilizing both the symbolic properties of polynomial ODEs and also their computational tameness. To the best of our knowledge this is the first complete axiomatization for such robust safety properties allowing for arbitrary (bounded) semialgebraic sets as initial/post-conditions. Earlier works concerning axiomatizability either do not account for robustness  \cite{DBLP:journals/jacm/PlatzerQ25, DBLP:journals/jacm/PlatzerT20}, only concern completeness relative to a non-computable oracle \cite{DBLP:conf/lics/Platzer12a} or enforce topological restrictions that are non-inductive \cite{DBLP:journals/jacm/PlatzerQ25, DBLP:conf/ijcar/AbouElWafaP26}. Prior works on robust ODEs \cite{DBLP:journals/jsc/Ratschan02, Franek_Ratschan_Zgliczynski_2016, Gao_Kong_Chen_Clarke_2014} primarily concern the computability-theoretic properties of such robust ODEs, in particular the computability of their reachable sets and therefore do not yield logical axiomatizations. 
\\
{\bf{\emph{Computability of robust ODEs}}}: Previous work \cite{DBLP:conf/csl/BlancB24} has essentially shown that the reachability relation of a robust ODE is (Type-Two) computable. However, as alluded to in the introduction, Type-Two computability fails to distinguish a set $S \subseteq \R^n$ from its closure $\overline{S}$. Hence such techniques cannot handle robust safety as presented in this article which allows for safety with no positive separation between the initial conditions and the unsafe sets. I.e. It is possible for a set of initial conditions $I \subseteq \R^n$ to be robustly safe and for $\overline{I}$ to be unsafe. Notions of approximate decidability \cite{Gao_Kong_Clarke_2013-dreal, Gao_Kong_Chen_Clarke_2014} have also been developed utilizing computability-theoretic focused treatments of robustness, but suffer from the same problem where the safety of $I$ and $\overline{I}$ cannot be distinguished, thereby lacking the capability of supporting inductive arguments used in theorem proving. Earlier work \cite{DBLP:journals/tac/Ratschan18} also established converse theorems for safety/barrier certificates assuming the dynamics robustness of ODEs. 
\\
{\bf{\emph{Proof theory of ODEs}}}: The qualitative study of ODEs through their logical axiomatizations was initiated in earlier works \cite{DBLP:journals/jar/Platzer08, DBLP:conf/lics/Platzer12a}, proving properties of ODEs through a small set of core axioms such that every established property is accompanied by a symbolic proof in deductive logic that can be independently checked in theorem provers implementing this logic (e.g. KeYmaera~X \cite{DBLP:conf/cade/FultonMQVP15}). Prior works in this direction established completeness properties for differential invariants \cite{DBLP:journals/jacm/PlatzerT20} and open properties of initial-value problems (IVPs) where the set of initial conditions is compact \cite{DBLP:journals/jacm/PlatzerQ25}, in essence the completeness of state robustness for ODEs. Recently, such completeness properties \cite{DBLP:journals/jacm/PlatzerQ25} were further extended to reachability properties of hybrid systems \cite{DBLP:conf/ijcar/AbouElWafaP26} by identifying a fragment of $\dL$ such that every definable dynamical system is state robust by construction. Importantly, such results \cite{DBLP:journals/jacm/PlatzerQ25, DBLP:conf/ijcar/AbouElWafaP26} are about state robustness and do not concern \emph{inductive} notions of \emph{robust safety}, and therefore do not apply in the setting of this article which concerns the axiomatization of topological robustness. This article can be seen as a generalization of completeness for (compact) IVPs to robust ODEs where the initial/post-conditions are no longer topologically constrained to compact/open sets. Theorem \ref{thm: compact IVP constrained safety completeness} is also a direct generalization of completeness of safety for compact IVPs that allows for unbounded domain constraints \cite{DBLP:journals/jacm/PlatzerQ25, DBLP:conf/ijcar/AbouElWafaP26}. Furthermore, the results of this article can also be seen as a generalization of completeness for differential invariants \cite{DBLP:journals/jacm/PlatzerT20} to completeness of \emph{robust} loop invariants for \emph{hybrid} dynamical systems. A notion that is similar in spirit to robust safety (safety modulo null sets) \cite{DBLP:conf/cade/CordwellP19} has also been studied in earlier work, but does not establish completeness results and is primarily focused on establishing some sound axiomatization.  

\section{Preliminaries}
\subsection{Differential Dynamic Logic}
\label{sec: prelim_dL}
This section provides a brief review of differential dynamic logic ($\dL$) and its axiomatization, focusing on the continuous fragment of $\dL$ \cite{DBLP:conf/lics/Platzer12a}.

\subsubsection{Syntax}
Terms in $\dL$ are formed by the following grammar, where $\V$ denotes the set of all variables, $x \in \V$ is a variable and $c \in \Q$ is a rational constant. Equivalently, terms are (multivariate) polynomials over $\V$ with rational coefficients.
\[p, q\bebecomes x\alternative c\alternative p + q\alternative p\cdot q\]

$\dL$ formulas have the following grammar, where ${\sim} \in \{=, \neq, \geq, >, \leq, <\}$ is a comparison relation and $\alpha$ is a system of differential equations ($\dL$ allows for $\alpha$ to be from the more general class of \emph{hybrid programs} \cite{DBLP:journals/jar/Platzer08}, which is not needed here)
\begin{align*}
    \phi, \psi &\bebecomes p {\sim} q \alternative \phi \land \psi \alternative \phi \lor \psi \alternative \neg \phi \alternative \forall x \phi \alternative \exists x \phi \alternative 
\ddiamond{\alpha}{\phi} \alternative \dbox{\alpha}{\phi}\\
\alpha &\bebecomes \cdots \alternative x' = p(x) \& Q
\end{align*}
This paper only deals with the case $\alpha \equiv x' = p(x) \& Q$, where $x' = p(x)$ represents an autonomous system of ODEs $x_1' = p_1(x), \cdots, x_n' = p_n(x)$ and $x = (x_1, \cdots, x_n)$ is understood to be vectorial. $Q$ here refers to some $\folr$ formula known as the \emph{domain constraint}. Intuitively, this restricts the region for which the ODE $x' = p(x)$ is allowed to evolve. In this article, it is assumed without loss of generality that all ODEs include a clock variable $t' = 1$. 

Lastly, we state some conventions that are used throughout this paper. For terms and formulas that appear in contexts involving ODEs $x' = p(x)$, it is sometimes useful to restrict the variables that they can refer to. When such cases arise, we will indicate such free variables by explicitly writing them as arguments. For example, $t()$ means that the term $t$ cannot refer to any bound variable of the ODE $x' = p(x)$. In contrast, $Q(x)$ (or just $Q$) indicates that all the variables may be referred to as free variables. These variable dependencies can be made formal and rigorous through $\dL$'s uniform substitution calculus \cite{DBLP:journals/jar/Platzer17}.

\subsubsection{Semantics}
A state $\omega$ is a mapping $\omega : \V \to \R$ that assigns a value to every variable. We denote $\S$ as the set of all such states. For a term $p$, its semantics in state $\omega \in \S$ written as $\eval{p}$ is the real value obtained by evaluating the polynomial $p$ at the state $\omega$. For a $\dL$ formula $\phi$, its semantics $\eval{\phi}$ is defined to be the set of all states $\omega \in \S$ such that $\omega \models \phi$, i.e. the formula $\phi$ is true in $\omega$. The semantics of first-order logical connectives are defined as usual, e.g. $\eval{\phi \lor \psi} = \eval{\phi} \cup \eval{\psi}$. For ODEs $\alpha \equiv x' = p(x) \& Q$, the semantics for $\dbox{\alpha}{\phi}$ and $\ddiamond{\alpha}{\phi}$ are defined as follows. For the given ODE $\alpha$ and any state $\omega \in \S$, let $\Phi_\omega : [0, T) \to \S$ be the solution to $x' = p(x)$ extended maximally to the right with $0 < T \leq \infty$ and $\Phi_\omega(0) = \omega$. We then have:
\begin{align*}
    &\omega \in \eval{\dbox{\alpha}{\phi}} \text{ iff for all $0 \leq \tau < T$ such that $\Phi_\omega(\xi) \models Q$ for all $0 \leq \xi \leq \tau$, we have $\Phi_\omega(\tau) \models \phi$}\\
    &\omega \in \eval{\ddiamond{\alpha}{\phi}} \text{ iff there exists some $0 \leq \tau < T$ such that $\Phi_\omega(\xi) \models Q$ for all $0 \leq \xi \leq \tau$ and $\Phi_\omega(\tau) \models \phi$}
\end{align*}
Intuitively, the formula $\dbox{\alpha}{\phi}$ expresses a \emph{safety} property, that $\phi$ holds along all flows of the ODE $x' = p(x)$ that remain inside the domain constraint $Q$. Dually, the formula $\ddiamond{\alpha}{\phi}$ expresses a \emph{liveness} property, that there is some flow along $x' = p(x)$ staying within $Q$ eventually reaching a state where $\phi$ is true. 

Finally, a formula $\phi$ is said to be valid if $\eval{\phi} = \S$, i.e. it is true in all states. For a formula $I$, we say it is a \emph{differential invariant} of the ODE $x' = p(x) \& Q$ if the formula $I \rightarrow [x' = p(x) \& Q]I$ is valid. One important fact is that $\dL$ is (effectively) complete for differential invariants in $\folr$ \cite{DBLP:journals/jacm/PlatzerT20}. In other words, if $I$ is a differential invariant, then one can effectively find a syntactic proof of $I \rightarrow [x' = p(x) \& Q]I$.

\subsubsection{Proof calculus}
The derivations in this paper are presented in a standard, classical sequent calculus with the usual rules for manipulating logical connectives and sequents. For a \emph{sequent} $\Gamma \vdash \phi$, its semantics is equivalent to the formula $\left(\bigwedge_{\psi \in \Gamma} \psi\right) \rightarrow \phi$, and the sequent is called valid if its corresponding formula is valid. For a sequent $\Gamma \vdash \phi$, formulas $\Gamma$ are called antecedents, and $\phi$ the succedent. Completed proof branches are marked with $*$ in a sequent proof, and since $\R$ has a decidable theory via quantifier elimination \cite{Tarski_1948}, statements that follow from real arithmetic are proven with the rule \irref{qear}. An axiom (schema) is called \emph{sound} iff all of its instances are valid, and a proof rule is sound if the validity of all its premises entails the validity of its conclusion. Axioms and proof rules are \emph{derivable} if they can be proven from $\dL$ axioms and proof rules via the aforementioned sequent calculus. Derivable axioms are automatically sound due to the soundness of $\dL$'s axiomatization \cite{DBLP:journals/jar/Platzer08, DBLP:journals/jacm/PlatzerT20}.

This article uses a fragment of the base axiomatization of $\dL$ \cite{DBLP:conf/lics/Platzer12b} (focusing on the continuous case) along with an extended axiomatization developed in prior works used to handle ODE invariants and liveness properties \cite{DBLP:journals/jacm/PlatzerT20, DBLP:journals/fac/TanP21}. A complete list of the axioms used is provided in Appendix \ref{sec: dL axiomatization}.

\subsection{Computability and Computable analysis}
\label{sec: Computable analysis prelims}
The completeness properties established in this article are \emph{effective}, there is a direct (computable) correspondence between the valid formulas and their proofs. That is, there is a computable algorithm taking valid formulas as inputs and outputting corresponding proofs in $\dL$. To achieve the desired completeness results effectively, it is necessary to utilize the computability-theoretic properties of ODEs framed in the language of \emph{computable analysis}. The following provides the required background on computable analysis, under the standard framework of \emph{Type Two Theory of Effectivity} (TTE) \cite{Weihrauch_2000}.  

\begin{definition}[Name]
    Let $x \in \R$ be any real number, a \emph{name} for $x$ is a sequence of rationals $(q_i)_i \subseteq \Q$ such that
    \[\boldsymbol{\forall} i \in \N~(\abs{q_i - x} < 2^{-i})\]
    This definition naturally extends to $\R^n$ by requiring names to reside in $\Q^n$ and using the standard Euclidean norm. For $x \in \R^n$, we denote the set of all names of $x$ as $\Gamma(x)$.
\end{definition}

For a real number $x \in \R^n$, its names act as the ``descriptions'' of $x$. A real $x$ is defined to be computable if it has a computable name. 

\begin{definition}[Type-two computable number]
    \label{def: computable real}
    Let $x \in \R^n$ be any real number, $x$ is \emph{Type-Two computable} if it has a computable name. That is, there is some computable sequence $(q_i)_i \subseteq \Q^n$ that is a name for $x$. 
\end{definition}

Intuitively, a real $x \in \R^n$ is (Type-Two) computable if and only if it can be computably approximated by a sequence of vectors of rational numbers with a computable rate of convergence. From now on, whenever we refer to the computability of numbers in $\R^n$, we mean Type-Two computability. 

\begin{definition}
An \emph{oracle machine $M$} is a Turing machine that allows for an additional one-way read-only input tape that represents some input oracle used. The machine is allowed to read this input tape up to arbitrary, but finite, lengths.
\end{definition}

Oracle machines are Turing machines with some access to outside information, the ``oracle'' input tape. The machine may use any finite amount of information on this tape. For an oracle machine $M$, and an infinite binary sequence $\rho \in 2^{\omega}$, $M^\rho$ denotes the oracle machine $M$ with oracle $\rho$. 

The following definition provides a notion of computability on the closed subsets of $\R^n$.

\begin{definition}[{\cite[Corollary~5.1.8]{Weihrauch_2000}}]
    A closed subset $E \subseteq \R^n$ is \emph{computable} if its corresponding distance function $x \mapsto \inf_{y \in E} \norm{x - y}$ is computable.
\end{definition}

It can be easily seen that every $\folr$ definable closed set is computable.

\begin{theorem}[\cite{DBLP:journals/jacm/PlatzerQ25}]
    \label{thm: definable sets are computable}
    If $E \subseteq \R^n$ is a closed subset defined by the $\folr$ formula $\phi(x)$, then it is a computable closed set and its distance function is computable uniformly in $\phi(x)$.
\end{theorem}

The following definition relates the use of oracle machines to computable functions in TTE.

\begin{definition}[Computable function]
    A function $f : E \subseteq \R^n \to \R^m$ with $E$ a computable closed set is \emph{computable} if there is some oracle machine $M$ such that 
    \[\boldsymbol{\forall}x \in E~\boldsymbol{\forall} \rho \in \Gamma(x)~((M^\rho(i))_i \in \Gamma(f(x)))\]
    I.e. $M$ maps names of $x$ to names of $f(x)$ for all $x \in E$. 
\end{definition}

A useful result of computable analysis is that the classical extreme value theorem holds computably \cite[Corollary~6.2.5]{Weihrauch_2000}.

\begin{theorem}[Computable extreme value theorem {\cite[Corollary~6.2.5]{Weihrauch_2000}}]
    \label{thm: computable extreme value theorem}
    Let $f : K \to \R$ be a computable function on the compact set $K \subset \R^n$ defined by some $\folr$ formula $\phi(x)$. Then $\max_{x \in K} (f(x))$ and $\min_{x \in K} (f(x))$ are uniformly computable in $f, \phi(x)$. 
\end{theorem}

\subsection{Subanalytic Geometry}
This section provides the results needed from subanalytic geometry \cite{MR972342}.
\begin{definition}[Semianalytic set \cite{MR972342, Kosiba_2023}]
    A set $A \subseteq \R^n$ is called \emph{semianalytic} if for all $x \in \R^n$ there is a neighborhood $x \in U$ and real analytic functions $f_i, g_{i, j}$ in $U$ such that 
    \[A \cap U = \bigcup^p_{i}\bigcap^q_j \{x \in U ~\vert~ f_i(x) = 0, g_{ij}(x) > 0\}\]
\end{definition}

\begin{definition}[Subanalytic set \cite{MR972342, Kosiba_2023}]
    A set $E \subseteq \R^n$ is \emph{subanalytic} if for all $x \in \R^n$, there is a neighborhood $x \in U$ such that $E \cap U = \pi(A)$, where $\pi : \R^{n + m} \to \R^n$ is the natural projection and $A \subseteq \R^{n + m}$ is a semianalytic set relatively compact in $\R^{n + m}$. 
\end{definition}

The family of subanalytic sets is closed under the standard Boolean operations of unions, intersections and complements \cite{MR972342}. 

\begin{definition}[Subanalytic function \cite{Kosiba_2023}]
    A function $f : U \to \R$ where $U \subseteq \R^n$ is a subanalytic set is subanalytic if its graph is a subanalytic set in $\R^{n + 1}$. 
\end{definition}
In particular, semianalytic sets are subanalytic and real analytic functions are subanalytic functions when restricted to compact subanalytic subsets of their domain. It then follows that taking the pointwise min-max of analytic functions results in subanalytic functions.

\begin{lemma}
    \label{lemma: pointwise min subanalytic}
    Let $K$ be a compact semialgebraic set and $(f_{i, j})_{i \leq n, j \leq m} : U \to \R$ be a finite collection of real analytic functions with $K \subset U$, then the function $h \coloneqq \min_{i \leq n} \max_{j \leq m}  f_{i, j} $ is subanalytic on $K$. 
\end{lemma}
\begin{proof}
    For $i \leq n$, denote $g_i \coloneqq \max_{j \leq m}f_{i, j}$, note that the graphs of such $g_i$'s can be written as 
    \[\graph(g_i) = \bigcup_{j \leq m} \graph(f_{i, j}) \cap \bigcap_{j \leq m} \{(x, y) ~\vert~ y - f_{i, j}(x) \geq 0\}\]
    which is a Boolean combination of subanalytic sets as the functions $f_{i, j}$ are real analytic on some neighborhood of $K$ and therefore subanalytic. Similarly, we have $h(x) = y$ if and only if
    \[(x, y) \in \bigcup_{i \leq n} \graph(g_i) \cap \bigcap_{i \leq n} \{(x, y) ~\vert~ g_i(x) - y \geq 0\}\]
    For a fixed $i \leq n$, further note that $g_i(x) - y \geq 0$ holds if and only if
    \[(x, y) \in \bigcup_{j \leq m} \{(x, y) | f_{i, j}(x) - y \geq 0\}\]
    Chaining these together we obtain
    \[\graph(h) = \bigcup_{i \leq n} \graph(g_i) \cap \bigcap_{i \leq n}\bigcup_{j \leq m} \{(x, y) | f_{i, j}(x) - y \geq 0\}\]
    which is a Boolean combination of subanalytic sets, thus $h : K \to \R$ is a subanalytic function. 
\end{proof}

The following lemma shows that the distance of an analytic function to a semialgebraic set is subanalytic, a regularity result. This result will be useful in establishing local progression of such distance functions, its proof is provided for completeness. 

\begin{lemma}
    \label{lemma: distance function is subanalytic}
    Let $K$ be a compact semialgebraic set, $f : U \to \R^n$ be a real analytic function with $K \subset U$, and $\emptyset \neq Q \subseteq \R^n$ be a semialgebraic set, then the distance function 
    \[x \mapsto d(f(x), Q)\]
    for $x \in K$ is subanalytic.
\end{lemma}

\begin{proof}
    Note that since $Q$ is semialgebraic, its distance function is also semialgebraic and therefore has a semialgebraic graph \cite{Bochnak_Coste_Roy_2013}. I.e. 
    \[\graph(d(\cdot, Q)) \equiv \left\{(x, t) \in \R^{n + 1} ~\biggr\vert~ \bigwedge_i\bigvee_j p_{i, j}(x, t) \succeq 0\right\}\]
    where each $p_{i, j}$ is a polynomial and $\succeq~\in \rels$. Hence, the graph of $x \mapsto d(f(x), Q)$ is
    \[\left\{(x, t) \in K \times \R ~\biggr\vert~ \bigwedge_i\bigvee_j p_{i, j}(f(x), t) \succeq 0\right\} = \bigcap_i\bigcup_j \left\{(x, t) \in K \times \R ~\vert~ p_{i, j}(f(x), t) \succeq 0\right\}\]
    which is a Boolean combination of semianalytic sets and therefore subanalytic, as desired. 
\end{proof}

The following variant of \L{}ojasiewicz's inequality for continuous subanalytic functions is needed. Intuitively, it establishes that continuous subanalytic functions admit a quantitative lower bound on its growth. 

\begin{theorem}[\L{}ojasiewicz's inequality {\cite[Proposition~3.17]{denkowska2008ensembles}}]
    \label{thm: subanalytic inequality}
    Let $f, g : K \to \R$ be continuous subanalytic functions on the compact set $K \subset \R^n$ such that $f^{-1}(0) \subseteq g^{-1}(0)$. Then there exist constants $C, N > 0$ such that 
    \[\abs{f(x)} \geq C\abs{g(x)}^N\]
    for all $x \in K$. 
\end{theorem}

\section{Unconstrained Safety}
\label{sec: unconstrained safety}
This section first establishes completeness and decidability results for safety problems without the presence of domain constraints, which are later generalized in Section \ref{sec: constrained safety} to allow for domain constraints. 
\begin{definition}[Conventions]
    The following standard conventions/results will be used throughout the article.
    \begin{itemize}
        \item $\folr$ denotes the set of all first-order formulas of real arithmetic without parameters. 
        \item A formula $S \in \folr$ is also identified with the set it defines (i.e. $\eval{S} = \{x \in \R^n ~\vert~ \R \models S(x)\}$).
        \item By standard results in real algebraic geometry \cite{Bochnak_Coste_Roy_2013}, it is well-known that every open (closed) semialgebraic set $S$ can be defined only using the atomic relation $>$ ($\geq$) in the following form (likewise for $\geq$)
        \[S \equiv \bigwedge_i \bigvee_j e_{i, j} > 0\]
        This article therefore assumes that definitions of open (closed) semialgebraic sets are in this form.
        \item For a formula $S \in \folr$, $S^o$ denotes a formula defining its topological interior (hence only uses $>$), and $\overline{S}$ a formula defining its topological closure (hence only uses $\geq$). 
        \item $\Q[x]$ denotes the ring of (multivariate) polynomials with rational coefficients in the variables $x = (x_1, \dots, x_n)$. 
        \item For a term $e(x)$ and a vector field $x' = p(x)$ with $p \in \Q[x]$, the $n$-th Lie derivative of $e(x)$ is denoted $\lied[n]{\genDE{x}}{e}$, defined inductively using the gradient $\nabla$ via
        \begin{align*}
            &\lied[0]{\genDE{x}}{e} = e\\
            &\lied[n + 1]{\genDE{x}}{e} = \nabla (\lied[n]{\genDE{x}}{e}) \cdot p(x)
        \end{align*}
        \item $d(x, y)$ always denotes the standard Euclidean distance $\norm{x - y}$ between $x, y \in \R^n$, for $A, B \subseteq \R^n$ we define $d(A, B) \coloneqq \inf_{x \in A, y \in B} d(x, y)$. This article uses the convention that the distance to an empty set is $\infty$, i.e. $d(x, \emptyset) = \infty$ for all $x$. 
    \end{itemize}
\end{definition}

\begin{definition}[Reachable Set]
    For an $n$-dimensional ODE $x' = p(x)$, with initial condition $x(0) = x_0$, $p(x)$ continuous and time $T \in \R^+$, the corresponding reachable set at time $T$ is denoted by $R^T_{p}(x_0)$ where $y \in R^T_{p}(x_0)$ if and only if there exists $\phi \in C^1([0, T], \R^n)$ such that
    \begin{itemize}
        \item $\phi(0) = x_0$.
        \item $\phi(T) = y$.
        \item $\phi'(t) = p(\phi(t))$ for all $t \in [0, T]$. 
    \end{itemize}
    $R^{[t_1, t_2]}_{p}(x_0)$ denotes $\bigcup_{t \in [t_1, t_2]} R^{t}_{p}(x_0)$ and $R^T_{p}(A)$ denotes $\bigcup_{x_0 \in A} R^T_{p}(x_0)$ for $A \subseteq \R^n$.
\end{definition}

Note that by standard results on the uniqueness of solutions to ODEs it follows that $R^T_{p}(x_0)$ contains at most one element. 

\begin{definition}[Unconstrained Safety Problem]
    A \emph{(bounded, unconstrained) polynomial safety problem} is a tuple 
    \[(p, I, S, T) \in \Q[x] \times \folr \times \folr \times \Q^{+}\]
    where $p(x)$ represents the vector field of the ODE $x' = p(x)$, $I(x)$ defines the (bounded, semialgebraic) initial conditions, $S(x)$ defines the (bounded, semialgebraic) set of safe states and $T$ denotes the bounded time horizon $[0, T]$. Throughout this article, safety problems are always assumed to be bounded. 
\end{definition}

The following definition defines what it means for a safety problem to be (robustly) safe. Intuitively, a safety problem is \emph{robustly valid} if all flows progress into the \emph{interior} of the safety set (Figure \ref{fig: robust safe sim}). 

\begin{definition}[Robust Safety]
    \label{def: robust safety}
    Let $(p, I, S, T)$ be a (bounded, unconstrained) safety problem. It is said to be:
    \begin{itemize}
        \item Safe if $R^{[0, T]}_{p}(I) \subseteq S$.
        \item Robustly safe if $I \subseteq S$ and $R^{(0, T]}_{p}(\overline{I}) \subseteq S^o$, where $\overline{I}$ denotes the topological closure of $I$ and $S^o$ denotes the topological interior of $S$. 
    \end{itemize}
\end{definition}

\begin{figure}
    \centering
    \includegraphics[width=0.65\linewidth]{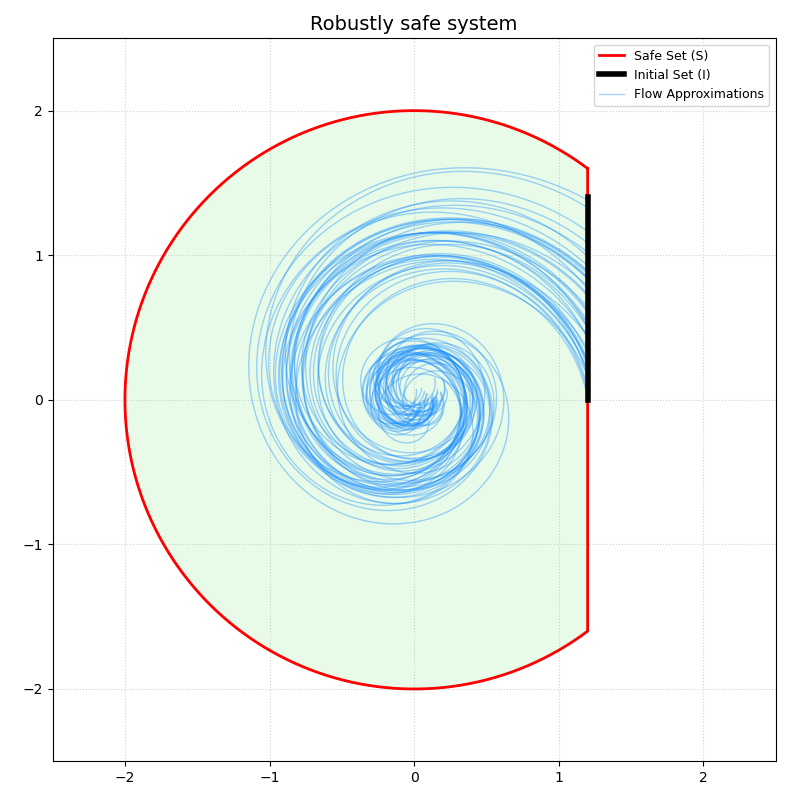}
    \captionof{figure}{Example of robust safety - flows progress into interior of safety set}
    \label{fig: robust safe sim}
\end{figure}

Importantly, note that the closure of the initial conditions $\overline{I}$ is \emph{not required} to be contained in the interior of the safety set $S^o$ at $t = 0$. This subtle yet crucial difference implies that it is possible for a safety problem to be robustly safe while satisfying $d(I, S^c) = 0$ (i.e. $\overline{I} \cap \overline{S^c} \neq \emptyset$), there need not be a (positive) separation between the initial condition $I$ and the unsafe set $S^c$. 

Since the reachable set $R^{[0, T]}_p(\overline{I})$ is computable for polynomial ODEs, semi-decidability/provability of $R^{[s, T]}_p(\overline{I}) \subseteq S^o$ for any rational $s \in (0, T)$ essentially follows from earlier works on compact initial-value problems \cite{DBLP:conf/csl/BlancB24,DBLP:journals/jacm/PlatzerQ25}. The main difficulty lies in proving/deciding the safety requirement $R^{(0, s]}_p(\overline{I}) \subseteq S^o$ for some initial segment $(0, s]$. Such properties are fundamentally challenging to handle using (Type-Two) computability-theoretic tools only, since computation with real-numbers is inexact, and any inexact approximation of $I$ could lead to $I \cap S^c \neq \emptyset$ thereby rendering the problem unsafe due to its approximation error. Consequently, we will establish such safety properties by leveraging the symbolic, qualitative properties of safety problems. The following definition is of central importance. 

\begin{definition}[Local Progression]
    Given an ODE $x' = p(x)$ and a set $S \subseteq \R^n$, $x_0$ \emph{locally progresses into $S$} \cite{DBLP:conf/emsoft/LiuZZ11, DBLP:journals/jacm/PlatzerT20} if there exists $s > 0$ such that 
    \[R^{(0, s]}_{p}(x_0) \subseteq S\]
    Note that $x_0 \in S$ is not required. Similarly, for a set of initial conditions $I \subseteq \R^n$, $I$ locally progresses into $S$ if there exists (a single) $s > 0$ such that 
    \[R^{(0, s]}_{p}(I) \subseteq S\]
\end{definition}

\begin{remark}
    Note that local progression for a set $I$ is \emph{not equivalent} to local progression of every point in $I$, since the duration of progression cannot be made uniform in general. Perhaps surprisingly, this fails to hold even when $I$ is compact as illustrated by the following example, as the maximal duration of progression $I \ni x_0 \mapsto \sup\{t > 0 | R^{(0, t]}_{p}(x_0) \subseteq S\}$ is not lower semicontinuous. 
\end{remark}

\begin{example}[Non-uniform local progression]
    This example shows that the local progression of every point in a set is in general \emph{strictly weaker} than the local progression of the entire set. I.e. there exist compact sets $I, S$ and ODE $x' = p(x)$ such that the following hold
    \begin{itemize}
        \item Every point in $I$ locally progresses into $S$:
         \[\forall x_0\in I~\exists s > 0~R^{(0, s]}_{p}(x_0) \subseteq S\]
        \item $I$ does not locally progress into $S$:
        \[\forall s > 0~\exists x_0 \in I~R^{(0, s]}_{p}(x_0) \nsubseteq S\]
    \end{itemize}
    The construction is as follows (on $\R^2$):
    \begin{align*}
        I(x, y) &\equiv 0 \leq x \leq 1 \land y = 0\\
        S(x, y) &\equiv (0 \leq x \leq 1 \land 0 \leq y \leq x) \lor  (-1 \leq x \leq 0 \land -1 \leq y \leq 0)\\
        (x', y') &= (-1, x)
    \end{align*}
    It can be observed that as $x_0 \to 0$, the time it takes for the corresponding flow $\phi((x_0, 0), t)$ to exit the set $S$ approaches $0$. However, the exit time at the origin $(0, 0)$ is exactly $1$, hence every point in $I$ locally progresses into $S$, but the set $I$ \emph{does not locally progress} into $S$, as no uniform lower bound on the exit times exists (Figure \ref{fig: non_uni_prog}).
    \begin{figure}
        \centering
        \begin{subfigure}[b]{0.6\textwidth}
            \centering
            \includegraphics[width=\textwidth]{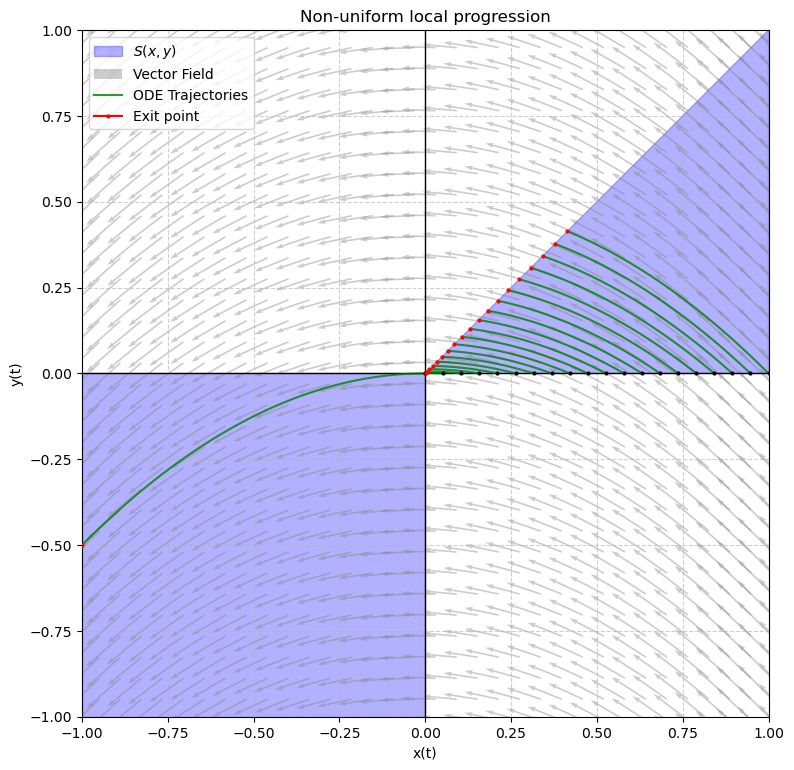}
        \end{subfigure}
        \caption{Non-uniform local progression.}
        \label{fig: non_uni_prog}
    \end{figure}
\end{example}

The notion of local progression for points has been previously studied \cite{DBLP:conf/emsoft/LiuZZ11,DBLP:journals/jacm/PlatzerT20} to prove that global (semialgebraic) differential invariants are decidable/complete for polynomial ODEs, whereas local progression for sets remains unexplored. Interestingly, this article shows that local progression for sets provides key insight into establishing decidability/completeness properties for robust safety. The following definition provides a quantitative measure of the ``degree of truth'' of a $\folr$ formula. 

\begin{definition}[Degree of Truth \cite{DBLP:journals/jsc/Ratschan02}]
    Let $\phi(\vec{x}) \in \folr$ be formula with $n$ free variables, its \emph{degree of truth function} $\theta_\phi : \R^n \to \R$ is defined as follows (note that $\theta_\phi$ is a constant when $n = 0$):
    \begin{itemize}
        \item By quantifier elimination, assume without loss of generality that $\phi(\vec{x})$ contains no quantifiers and is in conjunctive normal form:
        \[\phi(\vec{x}) \equiv \bigwedge_{i}\bigvee_{j} e_{i, j} \succeq 0\]
        where ${\succeq} \in \rels$. 
        \item Define $\theta_\phi(\vec{x})$ as:
        \[\theta_\phi(\vec{x}) = \min_{i}\{\max_{j}\{e_{i, j}(\vec{x})\}\}\]
    \end{itemize}
\end{definition}

The following observation is immediate by construction.
\begin{proposition}
    \label{prop: degree of truth characterization}
    Let $\phi(\vec{x}) \in \folr$ be a formula and $\theta_\phi$ its truth function. Then for all $\vec{z} \in \R^n$, the following hold
    \begin{itemize}
        \item $\theta_\phi(\vec{z}) > 0 \implies \vec{z} \in \eval{\phi}$. In particular, if $\phi$ is of the form
        \[\phi = \bigwedge_{i}\bigvee_j e_{i, j} > 0\]
        then $\theta_\phi(\vec{z}) > 0 \iff \vec{z} \in \eval{\phi}$. 
        \item $\theta_\phi(\vec{z}) < 0 \implies \vec{z} \notin \eval{\phi}$. In particular, if $\phi$ is of the form
        \[\phi = \bigwedge_{i}\bigvee_j e_{i, j} \geq 0\]
        then $\theta_\phi(\vec{z}) \geq 0 \iff \vec{z} \in \eval{\phi}$. 
    \end{itemize}
\end{proposition}

The following theorem provides an important characterization for local progression of compact sets into open sets. 

\begin{theorem}[Characterization of Compact-Open Local Progression]
        \label{thm: compact progression characterization}
    Let $p \in \Q[x]$ be a polynomial vector field, $K, S \in \folr$ be semialgebraic sets with $K$ compact and $S$ open.  Then the following are equivalent. 
    \begin{enumerate}
        \item  $K$ locally progresses into $S$ under $x' = p(x)$.
        \item There exists $s > 0$ such that the following hold:
            \begin{itemize}
                \item The flow $\phi(x, t)$ of $x' = p(x)$ is well defined on $K \times [0, s]$ (i.e. does not blow-up). 
                \item There exists some $k \in \N$ such that for all $(x_0, t) \in K \times [0, s]$, $\theta_S(\phi(x_0, t)) \geq t^k$, where $\theta_S$ is the degree of truth function of $S$. 
            \end{itemize}
        \item There exists $s > 0$ such that the following hold:
            \begin{itemize}
                \item The flow $\phi(x, t)$ of $x' = p(x)$ is well defined on $K \times [0, s]$ (i.e. does not blow-up). 
                \item There exists some $k \in \N$ such that for all $(x_0, t) \in K \times [0, s]$, $d(\phi(x_0, t), S^c) \geq t^k$, where $d(\cdot, \cdot)$ is the standard (Euclidean) distance function. 
            \end{itemize}
    \end{enumerate}
\end{theorem}

\begin{proof}
    The implications $(3), (2) \implies (1)$ are immediate from Proposition \ref{prop: degree of truth characterization}, so it suffices to establish $(1) \implies (2), (3)$. We first prove $(1) \implies (2)$. Since the right interval of maximal existence is lower semi-continuous \cite[Theorem~3.2]{Hartman_2002} in the space variables, we may assume without loss of generality that the flow $\phi$ of $x' = p(x)$ is well-defined on $K \times [0, s]$ for some sufficiently small $s > 0$. Furthermore, since $p(x)$ is analytic, the flow $\phi$ is also analytic (for sufficiently small $s$) on some neighborhood of $K \times [0, s]$, consequently $\theta_S(\phi(x_0, t)) : K \times [0, s] \to \R$ is subanalytic by Lemma \ref{lemma: pointwise min subanalytic} since $\theta_S$ is the pointwise min-max of analytic functions and $K \times [0, s]$ is compact semialgebraic. Finally, further pick $s > 0$ sufficiently small such that $R^{(0, s]}_p(K) \subseteq S$ which is possible by (1). This construction of $s$ implies that for any $(x_0, t) \in K \times [0, s]$, $\theta_S$ evaluates to $0$ only when $t = 0$:
    \[\theta_S(\phi(x_0, t)) = 0 \implies \phi(x_0, t) \notin S \implies t = 0\]
    Since $K \times [0, s]$ is compact, \L{}ojasiewicz's inequality (Theorem \ref{thm: subanalytic inequality}) implies that there exists $C, N > 0$ such that 
    \[\abs{\theta_S(\phi(x_0, t))} \geq Ct^N\]
    for all $(x_0, t) \in K \times [0, s]$. By increasing $N$ if necessary, and noting that $\theta_S(\phi(x_0, t)) \geq 0$ by construction, it follows that $\theta_S(\phi(x_0, t)) \geq t^N$ holds for all $(x_0, t) \in K \times [0, s]$, as desired. The proof of $(1) \implies (3)$ is identical by noting that the function $(x_0, t) \mapsto d(\phi(x_0, t), S^c)$ is subanalytic (Lemma \ref{lemma: distance function is subanalytic}) and the same argument applies. 
\end{proof}

\begin{remark}
    Theorem \ref{thm: compact progression characterization} can be viewed as a statement concerning ``quantitative progression'', in the sense that if $K$ locally progresses into $S$, then it does so at (at least) some polynomial rate $t^k$. 
\end{remark}

Local progression is also naturally definable in $\dL$, a key component in the axiomatization of robust safety. 

\begin{proposition}[$\dL$-Definability of Local Progression]
    \label{prop: LP dL def}
    Let $K, S \in \folr$ and $p(x) \in \Q[x]$ a polynomial vector field, the local progression of $K$ into $S$ is equivalent to the validity of the following \dL formula (where $t$ is a fresh variable)
    \[\lp_p(K, S) \equiv \lexists{s {>} 0} \lforall{x} (K(x) \land t = 0 \rightarrow \dbox{x' = p(x), t' = 1 \& t \leq s}{(t > 0 \rightarrow S(x))})\]
\end{proposition}

\begin{proof}
    This follows directly by the semantics of $\dL$ \cite{DBLP:journals/jar/Platzer08}.
\end{proof}

A consequence of Theorem \ref{thm: compact progression characterization} is that $\dL$ is complete for such local progression properties for compact initial conditions progressing into open sets. 

\begin{theorem}[Completeness of Compact-Open Local Progression]
    \label{thm: LP completeness}
    $\dL$ is complete for compact progression into open sets. That is, for all $K, S \in \folr$ and $p \in \Q[x]$ such that $K$ is compact and $S$ is open, $\lp_p(K, S)$ is provable if and only if it is valid.
    \[\models \lp_p(K, S) \iff\, \vdash \lp_p(K, S)\]
    Furthermore, in either (hence both) case, a corresponding $s \in \Q^+$ can be computed such that $R^{(0, s]}_p(K) \subseteq S$ (and this is further $\dL$-provable). 
\end{theorem}

Intuitively, the proof of Theorem \ref{thm: LP completeness} follows from applying the equivalence $(1) \iff (2)$ in Theorem \ref{thm: compact progression characterization} and approximating the degree of truth function $\theta_S$ via Taylor polynomials. The following theorem establishing the provability of Taylor's Theorem in $\dL$ will be useful to complete the proof. 

\begin{theorem}[Provable Taylor bounds]
    \label{thm: provable taylor}
    Let $e(x)$ be a term, $Q(x, t) \in \folr$, $n \in \N$ and $T_n(x_0, t)$ denote the $n$-th Taylor approximant of $e$, i.e.
    \[T_n(x_0, t) = \sum_{i = 0}^n \frac{\lied[i]{\genDE{x}}{e}(x_0)}{i!}t^{i}\]
    Then the following axiom is derivable for all $n \in \N$.

    \begin{calculus}
        \cinferenceRule[taylor|TA]
        {Taylor approximants}
        {
         \linferenceRule[impll2]
         {x = x_0 \land t = 0 \land \dbox{x' = p(x), t' = 1 \& Q(x, t)}{\abs{\lied[n + 1]{\genDE{x}}{e}(x)}\leq M } }
         {\dbox{x' = p(x), t' = 1 \& Q(x, t)}{\abs{T_n(x_0, t) - e(x)} \leq \frac{Mt^{n + 1}}{(n + 1)!}}}
        }{}
    \end{calculus}
\end{theorem}

The proof of Theorem \ref{thm: provable taylor} is provided in Appendix \ref{sec: provable taylor approximation}. Note that the bound $M$ is a symbolic variable rather than a fixed rational constant. 

\begin{proof}[Proof of Theorem \ref{thm: LP completeness}]
    The converse implication follows from the soundness of $\dL$ \cite{DBLP:journals/jar/Platzer08}, hence it suffices to establish the forward implication. Assume that $\lp_p(K, S)$ is valid, Theorem \ref{thm: compact progression characterization} then implies the existence of some $N \in \N$ and $s \in \Q^+$ such that 
    \[E(x_0, t) \coloneq \theta_S(\phi(x_0, t)) \geq t^N\]
    holds for $(x_0, t) \in K \times [0, s]$ and the flow $\phi$ is well-defined on $K \times [0, s]$, further choose $R \in \Q^+$ large enough such that $\phi(K, [0, s]) \subset B(0, R)$. Since $S \in \folr$ defines an open set, recall that it is (without loss of generality) of the form 
    \[S = \bigwedge_i \bigvee_j e_{i, j} > 0\]
    By defining $S_i = \bigvee_j e_{i, j} > 0$ and noting that box modalities commute with conjunctions via axiom \irref{band}, it follows that conjunctions provably commute with local progression, i.e. the following is provable
    \[\vdash_{\dL} \lp_p(K, S) \leftrightarrow \bigwedge_i \lp_p(K, S_i)\]
    Thus, it suffices to establish completeness for each of the $S_i$'s, and we may assume without loss of generality that $S$ is of the form 
    \[S = \bigvee_{j} e_j > 0\]
    For $k \in \N$, denote by $E_k$ the corresponding function defined with the $k$-th order Taylor approximations, i.e.
    \[E_k(x_0, t) \equiv \max_{j} T^j_k(x_0, t)\]
    where $T^j_k$ is the $k$-th order Taylor approximation of $e_j(x_0, t)$ as defined in Theorem \ref{thm: provable taylor}. Now carry out an a priori unbounded search for a pair $(k, \tau) \in \N \times \Q^+$ such that the following conditions hold:
    \begin{enumerate}
        \item $E_k(x_0, t) \geq t^{k}$ for all $(x_0, t) \in K \times [0, \tau]$.
        \item $\frac{M}{(k + 1)!}\tau < \frac{1}{2}$ and $\phi(K, [0, \tau]) \subset B(0, R)$ where $M \coloneq \max_j \abs{\lied[k + 1]{\genDE{x}}{e_j}(x)}_{x \in B[0, R]}$.\\
        Note that this implies $\abs{E_k(x_0, t) - E(x_0, t)} \leq \frac{t^{k}}{2}$ for $(x_0, t) \in K \times [0, \tau]$. 
    \end{enumerate}
    It follows from Theorem \ref{thm: compact progression characterization} that this search is bounded, since taking $k \geq N + 2$ and $\tau \in \Q^+$ sufficiently small necessarily satisfies both conditions. Furthermore, as the conditions are $\folr$-definable, this search is effective and computes an interval $[0, \tau]$ on which $K$ locally progresses into $S$. 
    
    It remains to establish the provability claim and show that local progression on $[0, \tau]$ can be deductively established in $\dL$. The derivation first begins with standard arithmetic manipulations to replace $S$ with its degree of truth $\theta_S$, $\alpha \equiv x' = p(x), t' = 1$ abbreviates the ODE for brevity
    \begin{sequentdeduction}
        \linfer[K+qear]{
            {\lsequent{}{K(x) \land t = 0 \rightarrow \dbox{\alpha \& t \leq \tau}{\theta_S(x) \geq \sfrac{t^k}{2}}}}
        }
        {\lsequent{}{K(x) \land t = 0 \rightarrow \dbox{\alpha \& t \leq \tau}{(t > 0 \rightarrow S(x))}}}
    \end{sequentdeduction}
    where $(k, \tau)$ are the witnesses to the search above. The derivation continues by approximating $\theta_S(x)$ with $E_k(x_0, t)$. 
    \begin{sequentdeduction}
        \linfer[cut+existsl]{
            \linfer[K+qear]{
                \linfer[band]{
                    \linfer[dW]
                    {
                        \linfer[qear]{
                            \lclose
                        }
                        {\lsequent{K(x_0), 0 \leq t \leq \tau}{E_k(x_0, t) \geq t^k}}
                    }
                    {\lsequent{}{x = x_0 \land K(x_0) \land t = 0 \rightarrow \dbox{\alpha \& t \leq \tau}{E_{k}(x_0, t) \geq t^k}}}
                    &
                    \text{\textcircled{A}}
                }
                {\lsequent{}{x = x_0 \land K(x_0) \land t = 0 \rightarrow \dbox{\alpha \& t \leq \tau}{(E_{k}(x_0, t) \geq t^k \land \abs{E_k(x_0, t) - \theta_S(x)} \leq \sfrac{t^k}{2})}}}
            }
            {\lsequent{}{x = x_0 \land K(x_0) \land t = 0 \rightarrow \dbox{\alpha \& t \leq \tau}{\theta_S(x) \geq \sfrac{t^k}{2}}}}
        }
        {\lsequent{}{K(x) \land t = 0 \rightarrow \dbox{\alpha \& t \leq \tau}{\theta_S(x) \geq \sfrac{t^k}{2}}}}
    \end{sequentdeduction}
    where the left premise closes by condition $(1)$ in the search for $(k, \tau)$, and premise \textcircled{A} concerns the error bound $\abs{E_k(x_0, t) - \theta_S(x)} \leq \sfrac{t^k}{2}$. The derivation of this error bound first begins by cutting in the domain constraint $B(0, R)$ defined via $\norm{x}^2 < R^2$. 
    \begin{sequentdeduction}
        \linfer[dC]{
            \linfer[]{
                \lclose
            }
            {\lsequent{}{x = x_0 \land K(x_0) \land t = 0 \rightarrow \dbox{\alpha \& t \leq \tau}{\norm{x}^2 < R^2}}}
            &
            \text{\textcircled{B}}
        }
        {\lsequent{}{x = x_0 \land K(x_0) \land t = 0 \rightarrow \dbox{\alpha \& t \leq \tau}{\abs{E_k(x_0, t) - \theta_S(x)} \leq \sfrac{t^k}{2}}}}
    \end{sequentdeduction}
and
\[\text{\textcircled{B}} \equiv  {\lsequent{}{x = x_0 \land K(x_0) \land t = 0 \rightarrow \dbox{\alpha \& t \leq \tau \land \norm{x}^2 < R^2}{\abs{E_k(x_0, t) - \theta_S(x)} \leq \sfrac{t^k}{2}}}}\]
where the left premise closes automatically as $\dL$ is complete for (bounded) open safety properties \cite{DBLP:journals/jacm/PlatzerQ25}. Note that condition $(2)$ guarantees the existence of some $W \in \Q^+$ such that $\frac{W}{(k + 1)!}\tau < \frac{1}{2}$ and $\max_j \abs{\lied[k + 1]{\genDE{x}}{e_j}(x)}_{x \in B[0, R]} \leq W$, hence \irref{taylor} applied to all $e_j$'s will yield the desired bounds.

\begin{sequentdeduction}
    \linfer[implyr+taylor+K]{
        \linfer[dC+dW+qear]{
            \lclose
        }
        {\lsequent{x = x_0, K(x_0), t = 0}{\dbox{\alpha \& t \leq \tau \land \norm{x}^2 < R^2}{\sfrac{Wt^{k + 1}}{(k + 1)!} \leq \sfrac{t^k}{2}}}}
        &
        \text{\textcircled{I}}
    }
    {\lsequent{}{x = x_0 \land K(x_0) \land t = 0 \rightarrow \dbox{\alpha \& t \leq \tau \land \norm{x}^2 < R^2}{\abs{E_k(x_0, t) - \theta_S(x)} \leq \sfrac{t^k}{2}}}}
\end{sequentdeduction}
with 
\[\text{\textcircled{I}} \equiv {\lsequent{}{\dbox{\alpha \& t \leq \tau \land \norm{x}^2 < R^2}{\max_j \abs{\lied[k + 1]{\genDE{x}}{e_j}(x)} \leq W}}}\]
the left premise closes with real-arithmetic rule \irref{qear} (after cutting in $0 \leq t \leq \tau$) since $\frac{W}{(k + 1)!}\tau < \frac{1}{2}$ holds by choice of $W$, and similarly \textcircled{I} also closes by \irref{dW+qear} as the choice of $W$ guarantees that the inequality is satisfied for $x \in B[0, R]$. As all premises are proven, this completes the proof. 
\end{proof}

\subsection{Completeness}
This section establishes that $\dL$ is complete for robust safety, establishing an equivalence between robust safety and \emph{provable} robust safety in $\dL$. The following proposition shows that (robust) safety can be naturally defined in $\dL$.

\begin{proposition}
    \label{prop: dL safety}
    Let $(p, I, S, T)$ be a safety problem, the following hold:
    \begin{itemize}
        \item It is safe if and only if the following $\dL$-formula is valid 
        \[\safe(p, I, S, T) \equiv I \land t = 0 \rightarrow \dbox{x' = p(x), t' = 1 \& t \leq T}{S}\]
        \item It is robustly safe if and only if the following $\dL$-formula is valid
        \[\safer(p, I, S, T) \equiv \left(I \rightarrow S\right) \land (\overline{I} \land t = 0 \rightarrow \dbox{x' = p(x), t' = 1 \& t \leq T}{\left(t > 0 \rightarrow S^o\right)})\]
        where $S^o \in \folr$ defines the topological interior of $S$ and $\overline{I}$ the topological closure of $I$.
    \end{itemize}
\end{proposition}

\begin{proof}
    Follows directly from the semantics of $\dL$ \cite{DBLP:journals/jar/Platzer08}. 
\end{proof}

By Proposition \ref{prop: dL safety}, (robust) safety properties are definable in $\dL$. The following theorem shows that $\dL$ is complete for robust safety.

\begin{theorem}[Completeness of Unconstrained Robust Safety]
\label{thm: complete for unconstrained safety}
$\dL$ is complete for (unconstrained) robust safety. That is, for all safety problems $(p, I, S, T)$, provability and validity coincide:
\[\models \safer(p, I, S, T) \iff\, \vdash_{\dL} \safer(p, I, S, T)\]
\end{theorem}

Before proceeding to the proof, we first give a sketch of the proof strategy. From earlier works, it is known that the reachable set of a polynomial ODE is (Type-Two) computable, which allows one to check for safety on intervals of the form $[s, T]$ when $s > 0$. However, this parameter $s$ cannot be made symbolic, and numerical approximations alone cannot prove the desired claim at time $t = 0$ where it is possible to have $I \cap S^{c} \neq \emptyset$. To handle this, we compute some numerical value $s > 0$ satisfying the following properties:

\begin{enumerate}
    \item On the time interval $[0, s]$, the set $\overline{I}$ locally progresses into $S^o$, which is provable due to completeness of local progressions (Theorem \ref{thm: compact progression characterization}).
    
    \item On the time interval $[s, T]$, safety can be proven by appealing to numerical approximations and computability of reachable sets, of which $\dL$ is also complete for \cite{DBLP:journals/jacm/PlatzerQ25}. 
\end{enumerate}

\begin{proof}
    The ``$\leftarrow$'' direction directly follows from soundness of $\dL$, so it suffices to assume the validity of $\safer(p, I, S, T)$ and establish provability in $\dL$. Note that the initial condition $I \rightarrow S$ is a $\folr$ formula hence provable by axiom \irref{qear}. The nontrivial part of $\safer(p, I, S, T)$ will be proved independently on the time intervals $[0, s]$ and $[s, T]$, where $s \in \Q^+$ is some witness to the local progression of $\overline{I}$ into $S^o$ as computed by Theorem \ref{thm: LP completeness}. The provability of $\safer(p, I, S, s)$ for $t \in [0, s]$ then follows by the definition of local progression and construction of $s$, as $\overline{I}$ progressing into $S^o$ implies robust safety. It remains to establish robust safety on $[s, T]$ and prove the following formula
    \[\overline{I} \land t = 0 \rightarrow \dbox{x' = p(x), t' = 1 \& t \leq T}{\left(t \geq s \rightarrow S^o\right)}\]
    As $\safer(p, I, S, T)$ is assumed to be valid, this necessarily implies $R^{[s, T]}_p(\overline{I}) \subset S^o$. Since $S^o$ is open and $R^{[s, T]}_p(\overline{I})$ is compact, this further implies the existence of some $\eps \in \Q^+$ such that $B(R^{[s, T]}_p(\overline{I}), \eps) \subseteq S^o$. By completeness of $\dL$ for approximations \cite{DBLP:journals/jacm/PlatzerQ25}, we may compute some polynomial $\eta \in \Q[x_0, t]$ such that $\eta(x_0, t)$ provably approximates the flow of $x' = p(x)$ with error $\sfrac{\eps}{3}$ on $\overline{I} \times [0, T]$, i.e. the following \dL formula is provable
    \[x = x_0 \land \overline{I}(x_0) \land t = 0 \rightarrow \dbox{x' = p(x), t' = 1 \& t \leq T}{\norm{\eta(x_0, t) - x}^2 < \sfrac{\eps^2}{3^2}}\]
    As $B(R^{[s, T]}_p(\overline{I}), \eps) \subseteq S^o$ and $\sfrac{\eps}{3} + \sfrac{\eps}{3} = \sfrac{2\eps}{3} < \eps$, the following is also provable using axioms \irref{K+qear}
    \[x = x_0 \land \overline{I}(x_0) \land t = 0 \rightarrow \dbox{x' = p(x), t' = 1 \& t \leq T}{(t \geq s \rightarrow B(\eta(x_0, t), \sfrac{\eps}{3}) \subseteq S^o)}\]
    Combining this formula with the previous then proves
    \begin{align}
        \label{eqn: t_geq_s_robust_safe}
        \overline{I}(x) \land t = 0 \rightarrow \dbox{x' = p(x), t' = 1 \& t \leq T}{(t \geq s \rightarrow S^o)}
    \end{align}
    as desired, establishing robust safety for $t \in [s, T]$. Finally, it remains to combine the proofs of robust safety on $[0, s]$ and $[s, T]$ into one proof for $\safer(p, I, S, T)$. This can be achieved by contradiction, if the system is not robustly safe, then it must exit $S^o$ at some point, i.e. the following formula is satisfiable
    \[\overline{I}(x) \land t = 0 \land \ddiamond{x' = p(x), t' = 1 \& t \leq T}{(t > 0 \land \neg S^o)}\]
    The proof then naturally splits into two cases corresponding to $t \leq s, t \geq s$ depending on when $x$ reaches $(S^o)^c$. The latter case results in a contradiction since robust safety for $t \geq s$ has been proven as formula (\ref{eqn: t_geq_s_robust_safe}). For the case of $t \leq s$, monotonicity of $t' = 1$ can be utilized to refine the domain constraint \cite[Lemma~5.9]{DBLP:journals/jacm/PlatzerQ25}, establishing the satisfiability of the following formula
    \[\overline{I}(x) \land t = 0 \land \ddiamond{x' = p(x), t' = 1 \& t \leq s}{(t > 0 \land \neg S^o)}\]
    However, this is precisely the negation of the (provably valid) formula $\safer(p, I, S, s)$, hence a contradiction. Since both cases result in a contradiction, this completes the proof. 
\end{proof}

\subsection{Approximate Decidability}
\label{sec: unconstrained decidability}

This section establishes approximate decidability results concerning robust safety. By Theorem \ref{thm: complete for unconstrained safety} and the fact that $\dL$'s axiomatization is computable, it naturally follows that robust safety is c.e. (computably enumerable). Recent works have shown that many c.e. problems that are undecidable in general are naturally ``approximately decidable'' \cite{Franek_Ratschan_Zgliczynski_2016, Gao_Avigad_Clarke_2012, DBLP:conf/csl/BlancB24}. There exists an algorithm that either correctly decides the truth of the problem, or the problem is sensitive to arbitrarily small perturbations. This section establishes results in a similar style for robust safety, that robust safety is decidable modulo small perturbations. It is important to note that such results \emph{do not follow} from earlier works as the initial condition $I$ is \emph{not} required to be positively separated away from the unsafe sets $S^c$.  

While the notion of robust safety for a safety problem $(p, I, S, T)$ is naturally defined, it is perhaps not as clear what the correct notion of ``perturbation'' to such problems should be. One natural interpretation is to note that robust safety is equivalent to requiring $d(R^{t}_p(\overline{I}), S^c) > 0$ for all $t \in (0, T]$. Taking this point of view, it is natural to consider perturbations of the form $d(R^{t}_p(I), S^c) \geq \delta(t)$ where $\delta(t)$ is the time-varying level of perturbation, leading to the following definition.

\begin{definition}[Perturbations of safety]
    \label{def: unconstrained_perturbation}
    Let $(p, I, S, T)$ be a safety problem. 
    \begin{itemize}
        \item A (computable) $\delta$-perturbation of the problem is some (computable) function $\delta \in C^0([0, T], \R^{\geq 0})$ such that $\delta(t) > 0$ for all $t \in (0, T]$.

        \item The problem is \emph{$\delta$-safe} if $I \subseteq S$ and $d(R^t_p(x), S^c) \geq \delta(t)$ for all $(x, t) \in \overline{I} \times (0, T]$. 
    \end{itemize}
\end{definition}

An attractive feature of Definition \ref{def: unconstrained_perturbation} is that a safety problem is robustly safe if and only if it is $\delta$-safe for some $\delta$-perturbation, hence such approximations ``converge'' as the level of perturbation $\delta$ tends to $0$.   

\begin{proposition}[Characterization of Robust Safety]
    \label{prop:robust safety char}
    A safety problem $(p, I, S, T)$ is robustly safe if and only if it is $\delta$-safe for some perturbation $\delta$.
\end{proposition}

\begin{proof}
    It is clear that $\delta$-safety implies robust safety, so it suffices to prove the converse. Suppose that $(p, I, S, T)$ is robustly safe, which in particular implies that $\overline{I}$ locally progresses into $S^o$. By Theorem \ref{thm: compact progression characterization}, there exists some $s \in (0, T]$ and $k \in \N$ such that $d(R^t_p(\overline{I}), S^c) \geq t^k$ for all $t \in (0, s]$. Furthermore, as the problem is robustly safe, there exists some positive $\eps > 0$ such that $d(R_p^{[s, T]}(\overline{I}), S^c) > \eps$. The $\delta$-perturbation defined as follows (which is positive for $t > 0$)
    \[\delta(t) = \min(\eps, t^k)\]
    is a valid bound for which the problem is $\delta$-safe.
\end{proof}

It is now possible to prove that robust safety is approximately decidable in the following sense.

\begin{theorem}[Approximate Decidability of Unconstrained Safety]
    \label{thm: approx decidability of unconstrained safety}
    There exists a computable algorithm such that given a safety problem $(p, I, S, T)$ and a computable $\delta$-perturbation, it correctly outputs one of the following:
    \begin{itemize}
        \item $(p, I, S, T)$ is robustly safe.
        \item $(p, I, S, T)$ is not $\delta$-safe. 
    \end{itemize}
    If both cases are true, then the algorithm can output either of the two. 
\end{theorem}

\begin{proof}
    Theorem \ref{thm: complete for unconstrained safety} implies that robust safety is c.e., so it suffices to show that not $\delta$-safe is c.e. as well. Indeed, we show that the following are equivalent
    \begin{enumerate}
        \item $(p, I, S, T)$ is not $\delta$-safe.
        \item $I \cap S^c \neq \emptyset$ or there exists $s \in (0, T)$ and $x \in \overline{I}$ such that $R^{s}_p(x) \neq \emptyset$ and $d(R^{s}_p(x), S^c) < \delta(s)$. 
    \end{enumerate}
    The implication $(2) \implies (1)$ is immediate, and $(1) \implies (2)$ follows from the fact that $\delta$-safety fails if and only if there exists some $s > 0$ with $d(R^{s}_p(x), S^c) - \delta(s) < 0$. Note that the ``no blow-up'' condition $R^{t}_p(x) \neq \emptyset$ in (2) is sound by the continuity of $d(R^{t}_p(x), S^c)$ and boundedness of $S$. Thus, it remains to establish that (2) is a c.e. relation. The condition $I \cap S^c \neq \emptyset$ is decidable and therefore c.e. For the second part, standard continuity arguments and density of the real-algebraic numbers $\tilde{\Q}$ in $\overline{I}$ (Lemma \ref{lem: real algebraic numbers are dense}) reduces the condition to:
    \[\exists~s \in \Q \cap (0, T)~\exists~x \in \overline{I} \cap \tilde{\Q} ~\left(s < t^+(x) \land d(R^{s}_p(x), S^c) - \delta(s) < 0 \right)\]
    where $t^+(x)$ denotes the (right) interval of maximal existence at $x$. Thus, the second part of (2) is equivalent to the condition above, which is c.e. as $s < t^+(x)$ is a uniformly c.e. relation \cite{Graca_Zhong_Buescu_2009}. As the flow $R^s_p(x)$ is computable for $s < t^+(x)$, the overall condition is c.e., thereby completing the proof as desired. 
\end{proof}

Similar to earlier works on approximate decidability \cite{Gao_Avigad_Clarke_2012,Franek_Ratschan_Zgliczynski_2016}, one can naturally define a notion of $\delta$-robustness for safety problems, and the robust safety of such problems will be \emph{exactly decidable}.

\begin{definition}[$\delta$-Robust Safety Problem]
    Let $(p, I, S, T)$ be a safety problem and $\delta(t)$ a computable perturbation. The safety problem is $\delta$-robust if exactly one of the following holds:
    \begin{itemize}
        \item $(p, I, S, T)$ is not robustly safe. 
        \item $(p, I, S, T)$ is $\delta$-safe. 
    \end{itemize}
\end{definition}

Intuitively, a safety problem is $\delta$-robust if its safety remains unchanged by a perturbation of level $\delta$. It follows from Theorem \ref{thm: approx decidability of unconstrained safety} that the robust safety of a $\delta$-robust safety problem can be decided from the computable perturbation $\delta(t)$, as the approximate algorithm either correctly decides that the problem is robustly safe, or decides that the problem is not $\delta$-safe, in which case it is also not robustly safe as it is assumed to be $\delta$-robust. 

\begin{corollary}[Decidability of $\delta$-Robust Safety]
    There exists a computable algorithm such that given a safety problem and a computable $\delta$-perturbation, correctly decides the robust safety of the problem provided that the problem is $\delta$-robust. 
\end{corollary}

\begin{proof}
    It suffices to run the algorithm given in Theorem \ref{thm: approx decidability of unconstrained safety}. If the algorithm decides that the safety problem is robustly safe, then we are done. Otherwise, the problem is not $\delta$-safe, since the problem is $\delta$-robust, it must therefore not be robustly safe. 
\end{proof}

A similar notion of $\delta$-robustness appeared in earlier works \cite{Gao_Avigad_Clarke_2012, Gao_Kong_Chen_Clarke_2014} concerning the $\delta$-decidability of $\R$. The results established above are a strict generalization in the context of safety problems for polynomial ODEs, as no positive separation between the initial conditions and the unsafe sets is required (Example \ref{eg: hybrid cart}).  

\section{Constrained Safety}
\label{sec: constrained safety}

This section develops strengthenings of the results presented in Section \ref{sec: unconstrained safety} to safety problems with domain constraints. Intuitively, domain constraints limit the region on which the flow can evolve, resulting in more general modeling capabilities and intricate dynamics. 

\begin{definition}[Constrained Reachable Set]
    \label{def: constrained reachable set}
    For an $n$-dimensional ODE $x' = p(x)$, with initial condition $x(0) = x_0$, time $T \in \R^+$ and domain constraint $Q \subseteq \R^n$, the corresponding constrained reachable set is denoted by $R^T_{p, Q}(x_0)$ where $y \in R^T_{p, Q}(x_0)$ if and only if there exists $\phi \in C^1([0, T], \R^n)$ such that
    \begin{itemize}
        \item $\phi(0) = x_0$.
        \item $\phi(T) = y$.
        \item $\phi'(t) = p(\phi(t))$ for all $t \in [0, T]$. 
        \item For all $t \in [0, T]$, $\phi(t) \in Q$. 
    \end{itemize}
    $R^{[t_1, t_2]}_{p, Q}(x_0)$ denotes $\bigcup_{t \in [t_1, t_2]} R^{t}_{p, Q}(x_0)$ and $R^T_{p, Q}(A)$ denotes $\bigcup_{x_0 \in A} R^T_{p, Q}(x_0)$ for $A \subseteq \R^n$.
\end{definition}

This article only considers the case where $Q$ is $\folr$-definable. Constrained safety problems are defined as follows.

\begin{definition}[Constrained Safety Problem]
    \label{def: constrained safety problem}
    A \emph{(bounded, constrained) polynomial safety problem} is a tuple 
    \[(p, I, S, Q, T) \in \Q[x] \times \folr \times \folr \times \folr \times \Q^{+}\]
    where $p(x)$ represents the vector field of the ODE $x' = p(x)$, $I(x)$ defines the (bounded, semialgebraic) initial conditions, $S(x)$ defines the (bounded, semialgebraic) set of safe states, $Q(x)$ defines the (semialgebraic) domain constraint and $T$ denotes the bounded time horizon $[0, T]$. 
\end{definition}

Note that every unconstrained safety problem is a constrained safety problem with $Q \equiv \top$. From now on safety problems are always assumed to be constrained. 

\begin{definition}[Constrained Robust Safety]
    \label{def: constrained robust safety}
    Let $(p, I, S, Q, T)$ be a safety problem. It is said to be:
    \begin{itemize}
        \item Safe if $R^{[0, T]}_{p, Q}(I) \subseteq S$.
        \item Robustly safe if $I \cap Q \subseteq S$, $R^{(0, T]}_{p, \overline{Q}}(\overline{I}) \subseteq S^o$ and $\overline{I}$ locally progresses into $Q^o$. 
    \end{itemize}
\end{definition}

That is, a constrained safety problem $(p, I, S, Q, T)$ is robustly safe if the initial state $\overline{I}$ can locally progress into the interior of the domain constraint $Q^o$, and the problem is furthermore robustly safe similar to Definition \ref{def: robust safety}. Intuitively, the condition of $\overline{I}$ locally progressing into $Q^o$ ensures that no initial point is stuck and constrained by $Q$ to not evolve at all. 

\subsection{Completeness}
Similar to Proposition \ref{prop: dL safety}, constrained safety properties are also definable in $\dL$.

\begin{proposition}
    \label{prop: constrained dL safety}
    Let $(p, I, S, Q, T)$ be a safety problem, the following hold:
    \begin{itemize}
        \item It is safe if and only if the following $\dL$-formula is valid 
        \[\safe(p, I, S, Q, T) \equiv I \land t = 0 \rightarrow \dbox{x' = p(x), t' = 1 \& t \leq T \land Q}{S}\]
        \item It is robustly safe if and only if the following $\dL$-formula is valid
        \begin{align*}
            \safer(p, I, S, Q, T) \equiv &\left(I \land Q \rightarrow S\right)\\
            \land &\lp_p(\overline{I}, Q^o)\\
            \land & \left(\overline{I} \land t = 0 \rightarrow \dbox{x' = p(x), t' = 1 \& t \leq T \land \overline{Q}}{\left(t > 0 \rightarrow S^o\right)}\right)         
        \end{align*}
    \end{itemize}
\end{proposition}

\begin{proof}
    This follows from Proposition \ref{prop: LP dL def} and Proposition \ref{prop: dL safety}.
\end{proof}

To establish completeness of constrained problems we first show that $\dL$ is complete for topologically open safety properties with compact initial conditions and closed domain constraints. This is a result of independent interest that generalizes earlier works \cite{DBLP:journals/jacm/PlatzerQ25, DBLP:conf/ijcar/AbouElWafaP26} on the completeness of initial value problems by allowing for (closed) unbounded domain constraints. 

\begin{theorem}[Completeness for Constrained Open Safety]
    \label{thm: compact IVP constrained safety completeness}
    $\dL$ is complete for safety problems with compact initial conditions, closed constraints and (bounded) open safe sets. That is, for all (bounded) safety problems $(p, I, S, Q, T)$ with $I$ compact, $S$ (bounded) open and $Q$ closed, provability and validity coincide:
    \[\models \safe(p, I, S, Q, T) \iff \vdash_{\dL} \safe(p, I, S, Q, T)\]
\end{theorem}

While Theorem \ref{thm: compact IVP constrained safety completeness} appears similar to its unconstrained counterpart \cite{DBLP:journals/jacm/PlatzerQ25}, it is considerably more subtle due to the possibility of finite time blow-ups. In the unconstrained case, since the set of safe states $S$ is bounded, the safety of the system necessarily implies that the overall flow must be bounded (hence well-defined) on $[0, T]$. In the constrained case, this is not necessarily the case, since safety could still hold even when the flow blows up on $[0, T]$ due to the constraint $Q$. Furthermore, it could be the case that the flow blows up for certain initial values in $I$, and exists on $[0, T]$ for other values. We first show that $I$ can be partitioned such that the blow-ups occur uniformly. The following well-known property of semialgebraic compact sets is needed, the proof is provided for completeness.

\begin{lemma}
    \label{lem: real algebraic numbers are dense}
    Let $K \in \folr$ be a compact semialgebraic set and let $\tilde{\Q}$ denote the real algebraic numbers. Then $\tilde{\Q} \cap K$ is dense in $K$.
\end{lemma}

\begin{proof}
    As $K$ is compact, for every $n \in \N$, there exists a finite cover of $K$ of the form $\{B(x_i, 2^{-n})~\vert~i \leq N, x_i \in K\}$. Since this is a first-order property and $\tilde{\Q}$ is an elementary substructure of the real closed field $\R$, there exists $\{y_i\}_{i \leq N} \subseteq \tilde{\Q} \cap K$ such that
    \[\tilde{\Q} \models K \subseteq \bigcup_{i \leq N} B(y_i, 2^{-n}) \]
    Again utilizing that $\tilde{\Q}$ is an elementary substructure of $\R$, we obtain 
    \[\R \models  K \subseteq \bigcup_{i \leq N} B(y_i, 2^{-n}) \]
    hence every $x \in K$ is of distance at most $2^{-n}$ to $\tilde{\Q} \cap K$ for all $n$, completing the proof. 
\end{proof}

We can now establish a ``uniform cover'' of $I$. For a vector field $x' = p(x)$, $t^{+}(x_0)$ denotes the (right) duration of maximal existence of the solution starting at $x = x_0$. 

\begin{lemma}[Uniform cover]
    \label{lem: uniform partition}
    Let $K \in \folr$ be a compact semialgebraic set, $p \in \Q[x]$ be a vector field, and $M, T \in \Q^+$. Then there exists a finite list of pairs $\{(x_i, r_i)\}_{i \leq N} \subseteq (\tilde{\Q} \cap K) \times \Q^+$ such that:
    \begin{itemize}
        \item $K \subseteq \bigcup_{i \leq N} B(x_i, r_i)$.
        \item For every $i \leq N$, at least one of the following is true:
        \begin{itemize}
            \item $t^+(z) > T$ for all $z \in B(x_i, r_i)$.
            \item There exists $0 < s < T$ such that for all $z \in B(x_i, r_i)$, $s < t^+(z)$ and $\norm{\phi(z, s)} > M$, where $\phi(z, s)$ denotes the flow of $x' = p(x)$.  
        \end{itemize}
    \end{itemize}
\end{lemma}

\begin{proof}
    Let $x \in K$ be arbitrary, we pick a corresponding radius $r\in \Q^+$ depending on the value of $t^+(x)$:
    \begin{enumerate}
        \item If $t^+(x) > T$, then since $t^+(x)$ is lower-semicontinuous in $x$, pick $r \in \Q^+$ sufficiently small such that every $z \in B(x, 2r)$ satisfies $t^+(z) > T$.
        \item If $t^+(x) \leq T$, then the flow starting at $x$ along $x' = p(x)$ blows up before reaching time $T$, hence $\norm{\phi(x, t)} \to \infty$ as $t \to t^+(x) \leq T$, so we may choose some $0 < s < t^+(x) \leq T$ such that $\norm{\phi(x, s)} > M$. Further choose the corresponding $r \in \Q^+$ to be small enough such that $t^{+}(z) > s$ for all $z \in B(x, 2r)$ and $\norm{\phi(z, s)} > M$, which is possible as $z \mapsto \phi(z, s)$ is continuous in $z$. 
    \end{enumerate}
    The resulting collection of open balls $\{B(x, r_x)\}_{x \in K}$ thus forms an open cover of $K$, by compactness we obtain a finite subcover $\{B(x_i, r_i)\}_{i \leq N}$ where each $x_i \in K$. To obtain a cover with points in $\tilde{\Q} \cap K$, simply replace each $(x_i, r_i)$ pair with $(y_i, \sfrac{4r_i}{3})$ where $y_i \in \tilde{\Q} \cap K$ and $d(y_i, x_i) < \sfrac{r_i}{3}$, which exists by density of $\tilde{\Q} \cap K$. To see that the resulting list $(y_i, \sfrac{4r_i}{3})$ still covers $K$ while also satisfying the desired conditions, note that 
    \[B(x_i, r_i) \subseteq B(y_i, \sfrac{4r_i}{3}) \subseteq B(x_i, 2r_i)\]
    hence the desired properties follow by construction, and the proof is complete. 
\end{proof}

This now allows us to prove Theorem \ref{thm: compact IVP constrained safety completeness}.

\begin{proof}[Proof of Theorem \ref{thm: compact IVP constrained safety completeness}]
    Let $M \in \Q^+$ be large enough such that $\norm{x}^2 \geq M^2 \rightarrow x \notin S$ holds, which is possible as $S$ is bounded. By Lemma \ref{lem: uniform partition}, we can find a $\folr$-definable cover $I = \bigcup_{i \leq n} I_i$ with this choice of $M$ that is uniform in blow-up times by setting $I_i \coloneqq I \cap B[x_i, r_i]$ (closed ball of radius $r_i$ around $x_i$)\footnote{Technically Lemma \ref{lem: uniform partition} requires using open balls, but one can always use $B(x_i, 2r_i) \supseteq B[x_i, r_i]$ and refine.}. Since each $I_i$ is $\folr$-definable, the following will be a valid (and therefore provable) $\folr$ formula
    \[I \rightarrow \bigvee_{i \leq n} I_i\]
    As such, it suffices to establish completeness by assuming at least one of the following is true:
    \begin{itemize}
        \item $t^+(z) > T$ for all $z \in I$.
        \item There exists some $s \in \Q^+$ with $0 < s < T$ such that $t^+(z) > s$ and $\norm{\phi(z, s)}^2 > M^2$ for all $z \in I$.
    \end{itemize}
    We establish the provability of $\safe(p, I, S, Q, T)$ for these cases separately.
    \begin{enumerate}
        \item Suppose that $t^+(z) > T$ holds for all $z \in I$, in this case the flow $\phi : I \times [0, T] \to \R^n$ is well-defined. Let $(\Phi_n)_n$ be a sequence of polynomial approximations to $\phi$ that converges uniformly. By earlier works \cite{DBLP:journals/jacm/PlatzerQ25}, $\dL$ is complete for convergence - for every $\eps \in \Q^+$, there computably exists some corresponding $n_\eps \in \N$ such that the following formula
        \[\text{ERROR}(\eps, k) \equiv I(x) \land x = x_0 \land t = 0 \rightarrow \dbox{x' = p(x), t' = 1 \& t \leq T}{\norm{x - \Phi_k(x_0, t)}^2 < \eps^2}\]
        is provable for all $k \geq n_\eps$. We now claim that there necessarily exists some pair $(\eps, k) \in \Q^+ \times \N$ such that:
        \begin{itemize}
            \item $\text{ERROR}(\eps, k)$ is provable.
            \item For all $z \in I$ and $\tau \in [0, T]$, if $B(\Phi_k(z, s), \eps) \cap Q \neq \emptyset$ for all $s \in [0, \tau]$, then $B(\Phi_k(z, \tau), \eps) \cap S^c = \emptyset$. 
        \end{itemize}
        Clearly if such a pair exists, then the provability of $\safe(p, I, S, Q, T)$ follows from the provability of $\text{ERROR}(\eps, k)$ and the fact that the second condition is $\folr$ expressible and therefore also provable when valid. Thus, it remains to establish the existence of some such pair. Suppose for the sake of contradiction that such pairs do not exist. In particular, this implies that for every $\eps \in \Q^+$, there exists some $(z_\eps, \tau_\eps) \in I \times [0, T]$ such that $B(\Phi_{n_\eps}(z_\eps, s), \eps) \cap Q \neq \emptyset$ for all $s \in [0, \tau_\eps]$, yet $B(\Phi_{n_\eps}(z_\eps, \tau_\eps), \eps) \cap S^c \neq \emptyset$. By compactness a subsequence converges, hence we may assume (up to taking a subsequence) that $(z_\eps, \tau_\eps) \to (z, \tau) \in I \times [0, T]$ as $\eps \to 0$. It suffices to show that $\phi(z, s) \in Q$ for all $s \in [0, \tau]$ and $\phi(z, \tau) \notin S$ which would yield a contradiction. Let $s \in [0, \tau]$ be arbitrary, the continuity of $\phi$ gives
        \[d(\phi(z, s), Q) = \lim_{\eps \to 0} d(\phi(z_\eps, s), Q)\]
        The proof continues by first establishing $\phi(z, s) \in Q$. As $Q$ is closed, it suffices to prove $d(\phi(z, s), Q) = 0$. Furthermore, we may assume without loss of generality that $s < \tau$ as $\phi(z, s)$ is continuous. Since $\tau_\eps \to \tau$ and $s < \tau$, this implies that $s < \tau_\eps$ as $\eps \to 0$. The definition of $\text{ERROR}(\eps, k)$ then implies
        \[\lim_{\eps \to 0} d(\phi(z_\eps, s), Q) \leq \lim_{\eps \to 0} 2\eps = 0 \]
        Since $Q$ is closed, this then implies that $\phi(z, s) \in Q$ for all $s \in [0, \tau]$. Similarly, we have
        \[d(\phi(z, \tau), S^c) = \lim_{\eps \to 0} d(\phi(z_\eps, \tau_\eps), S^c) = 0\]
        Because $S^c$ is closed, this implies $\phi(z, \tau) \in S^c$, contradicting the validity of $\safe(p, I, S, Q, T)$ as desired. 

        \item For the case where $\abs{\phi(z, s)} > M$ for all $z\in I$, first note that the validity of $\safe(p, I, S, Q, T)$ necessarily implies the validity of $\safe(p, I, S, Q, s)$ as $[0, s] \subseteq [0, T]$, and by part $(1)$ we have established that $\safe(p, I, S, Q, s)$ is provable. Furthermore, by an identical argument of passing to polynomial approximations, we see that the following is also provable 
        \begin{align}
            \label{eqn: cons_safe_pf_1}
            I(x) \land t = 0 \rightarrow \dbox{x' = p(x), t' = 1 \& t \leq T}{\left(t = s \rightarrow \norm{x}^2 > M^2\right)}
        \end{align}
        Thus, it suffices to show that $\safe(p, I, S, Q, T)$ is provable from the formula above and $\safe(p, I, S, Q, s)$. Indeed, by choice of $M$, the following is provable by applications of axioms \irref{qear+K}
        \begin{align}
            \label{eqn: cons_safe_pf_2}
            I(x) \land t = 0 \rightarrow \dbox{x' = p(x), t' = 1 \& t \leq T}{\left(\norm{x}^2 > M^2 \rightarrow \neg S\right)}
        \end{align}
        Combining (\ref{eqn: cons_safe_pf_1}) and (\ref{eqn: cons_safe_pf_2}) with axiom \irref{K} proves
        \[I(x) \land t = 0 \rightarrow \dbox{x' = p(x), t' = 1 \& t \leq T}{(t = s \rightarrow \neg S)}\]
        I.e. unconstrained safety does not hold at $t = s$. Since constrained safety $\safe(p, I, S, Q, s)$ is provable over the time interval $[0, s]$, it must necessarily be the case that the dynamics never reaches $t = s$ under the constraint $Q$. In other words, the following \dL formula is provable (by cutting in the domain constraint $t = s \rightarrow \neg S$) 
        \[I(x) \land t = 0 \rightarrow \dbox{x' = p(x), t' = 1 \& t \leq s \land Q}{t < s}\]
        An application of axiom \irref{enclosure} then relaxes the time constraint of $t \leq s$ to $t \leq T$, proving
        \[I(x) \land t = 0 \rightarrow \dbox{x' = p(x), t' = 1 \& t \leq T \land Q}{t < s}\]
        Thus, under the constraints $Q$ and $t \in [0, T]$, the dynamics can only evolve over the smaller time horizon $t \in [0, s]$. In other words, the safety of $t \in [0, T]$ is equivalent to the safety over $t \in [0, s]$ because $t$ cannot reach a state with $t \geq s$ under the constraint $Q$. The provability of constrained safety over $[0, s]$ defined by the \dL formula $\safe(p, I, S, Q, s)$ then implies the overall safety $\safe(p, I, S, Q, T)$ for $t \in [0, T]$ by the axiom \irref{dC}.
        \qedhere
    \end{enumerate}
\end{proof}

The following theorem establishes that $\dL$ is complete for the robust safety of constrained safety problems.

\begin{theorem}[Completeness of Constrained Robust Safety]
    \label{thm: complete for constrained safety}
    $\dL$ is complete for (constrained) robust safety of constrained safety problems. That is, for all constrained safety problems $(p, I, S, Q, T)$, provability and validity of robust safety coincide:
    \[\models \safer(p, I, S, Q, T) \iff \vdash_{\dL} \safer(p, I, S, Q, T)\]
\end{theorem}

\begin{proof}
    As \dL is sound, it suffices to assume validity and establish provability. Note that by Theorem \ref{thm: LP completeness}, $\dL$ is complete for (compact) local progressions into open sets, thus the provability of $\lp_p(\overline{I}, Q^o)$ follows. Furthermore, $I \land Q \to S$ is a valid $\folr$ formula and therefore provable as well, thus it remains to establish the provability of 
    \[\overline{I} \land t = 0 \rightarrow \dbox{x' = p(x), t' = 1 \& t \leq T \land \overline{Q}}{\left(t > 0 \rightarrow S^o\right)}\]
    We proceed similarly to Theorem \ref{thm: complete for unconstrained safety}, note that since $\overline{I}$ locally progresses into $Q^o$ \emph{and} $\safer(p, I, S, Q, T)$ is valid, it follows that $\overline{I}$ locally progresses into $S^o$. Thus, by Theorem \ref{thm: LP completeness}, there exists some $s \in \Q^+$ such that the following is provable
    \[\overline{I} \land t = 0 \rightarrow \dbox{x' = p(x), t' = 1 \& t \leq s}{(t > 0 \rightarrow S^o)}\]
    Hence, it suffices to establish the provability of the following formula, similar to the proof of Theorem \ref{thm: complete for unconstrained safety}.
    \[\overline{I} \land t = 0 \rightarrow \dbox{x' = p(x), t' = 1 \& t \leq T \land \overline{Q}}{(t \geq s \rightarrow S^o)}\]
    The (semantic) validity of this formula and the local progression on $[0, s]$ implies that there exists some provable approximation $\theta(x_0, t)$ of sufficiently small error $\eps > 0$ to the flow such that the following are valid \dL formulas 
    \begin{align*}
        \overline{I}(x) \land x = x_0 \land t = 0 &\rightarrow \dbox{x' = p(x), t' = 1 \& t \leq s}{\norm{x - \theta(x_0, t)}^2 \leq \eps^2}\\
        \overline{I}(x_0) \land \norm{x - \theta(x_0, t)}^2 \leq \eps^2 \land t = s &\rightarrow \dbox{x' = p(x), t' = 1 \& t \leq T \land \overline{Q}}{S^o}
    \end{align*}
    where the validity of the second formula follows from the fact that the image of the compact set $\overline{\eval{I}}$ under the flow map will also be compact. Indeed, suppose that the second formula is not valid, so for all $\eps > 0$ there exists a corresponding $(x_\eps, t_\eps)$ such that the following hold:
    \begin{itemize}
        \item $x_\eps \in B(R^s_p(\overline{I}), \eps)$
        \item $t_\eps \in [s, T]$
        \item $\forall \tau \in [0, t_\eps - s]~R^{\tau}_p(x_\eps) \in \overline{Q}$
        \item $R^{t_\eps - s}_p(x_\eps) \notin S^o$
    \end{itemize}
    By compactness, the pair sequence $(x_\eps, t_\eps)$ then converges to some $(x^*, t^*) \in R^s_p(\overline{I}) \times [s, T]$ as $\eps \to 0$. By continuity and openness of $S^o$, it follows that $R^{t^* - s}_p(x^*) \notin S^o$. $\overline{Q}$ being closed further implies that $\forall \tau \in [0, t^* - s]~R^{\tau}_p(x^*) \in \overline{Q}$, witnessing the property that $x^*$ is unsafe at time $t = t^*$, a contradiction. Hence the second formula is valid. 
    
    Note that the first formula is provable by construction of $\theta(x_0, t)$ \cite[Corollary~5.6]{DBLP:journals/jacm/PlatzerQ25}, and the second is provable by Theorem \ref{thm: compact IVP constrained safety completeness}. Hence, by the step extension axiom of \dL \cite[Theorem~5.7]{DBLP:journals/jacm/PlatzerQ25}, these formulas jointly prove the desired safety condition on $[s, T]$
    \[\overline{I} \land t = 0 \rightarrow \dbox{x' = p(x), t' = 1 \& t \leq T \land \overline{Q}}{(t \geq s \rightarrow S^o)}\]
    thereby completing the proof as desired. 
\end{proof}

\subsection{Approximate Decidability}
This section establishes analogous decidability results to those in Section \ref{sec: unconstrained decidability} for constrained safety problems. 

\begin{definition}[Perturbations of Constrained Safety]
    \label{def: constrained perturbation}
    Let $(p, I, S, Q, T)$ be a safety problem. 
    \begin{itemize}
        \item A (computable) $\delta$-perturbation of the problem is some (computable) function $\delta \in C^0([0, T], \R^{\geq 0})$ such that $\delta(t) > 0$ for all $t \in (0, T]$.

        \item The problem is \emph{$\delta$-safe} if $I \cap Q \subseteq S$ and $d(R^t_{p, B(Q, \delta)}(x), S^c) \geq \delta(t)$ for all $(x, t) \in \overline{I} \times (0, T]$, where the domain constraint $Q$ is relaxed to the time-varying set $B(Q, \delta(t))$ for all positive times, i.e. if $\phi(x, t)$ is the flow, then 
        \[y \in  R^t_{p, B(Q, \delta)}(x) \iff \phi(x, t) = y \land x \in \overline{Q} \land \forall s \in (0, t]~\phi(x, s) \in B(Q, \delta(s))\]
    \end{itemize}
\end{definition}

Intuitively, a $\delta$-perturbation of a constrained safety problem is a perturbation of the corresponding unconstrained problem together with widening the domain constraint by a time-varying factor of $\delta(t)$ at time $t$. The following theorem establishes approximate decidability for constrained safety problems with initial conditions progressing into the domain constraint. 

\begin{theorem}[Approximate Decidability of Constrained Safety]
    \label{thm: approx decidability of constrained safety}
    There exists a computable algorithm such that given a constrained safety problem $(p, I, S, Q, T)$ where $\overline{I}$ locally progresses into $Q^o$ and a computable $\delta$-perturbation, it correctly decides one of the following:
    \begin{itemize}
        \item $(p, I, S, Q, T)$ is robustly safe.
        \item $(p, I, S, Q, T)$ is not $\delta$-safe. 
    \end{itemize}
\end{theorem}

\begin{proof}
    Since $\dL$ is computably axiomatized and complete for robust safety of constrained problems, it follows that deciding $\safer(p, I, S, Q, T)$ is c.e., and it remains to show that deciding $(p, I, S, Q, T)$ to be not $\delta$-safe is also c.e. As $\lp_p(\overline{I}, Q^o)$ is valid by assumption, let $s \in \Q^+$ be small enough such that $R^{(0, s]}(\overline{I}) \subseteq Q^o$ and the flow is well-defined on $[0, s]$. We establish that the following are equivalent:
    \begin{enumerate}
        \item $(p, I, S, Q, T)$ is not $\delta$-safe.
        \item Either of the following hold:
            \begin{enumerate}
                \item $I \cap Q \cap S^c \neq \emptyset$.
                \item There exists $s_1 \in (0, s)$ such that 
                \[\min_{(x, t) \in \overline{I} \times [s_1, s]}(d(R^{t}_p(x), S^c) - \delta(t)) < 0\]
                \item There exists $z \in \tilde{\Q} \cap \overline{I}$, $s_1 \in \Q^+$ with $s < s_1 \leq T$ and $s_1 < t^+(z)$ such that the following are valid
                \begin{align*}
                    &\max_{t \in [s, s_1]} (d(R_p^t(z), Q) - \delta(t)) < 0\\
                    &d(R^{s_1}_p(z), S^c) < \delta(s_1)
                \end{align*}
            \end{enumerate}
    \end{enumerate}
    
    Note that condition $(2)$ is c.e. as extreme values of computable functions are computable over computably compact sets (Theorem \ref{thm: computable extreme value theorem}) and $s_1 < t^+(z)$ is a c.e. relation \cite{Graca_Zhong_Buescu_2009}, so it suffices to establish this equivalence. The implication $(2) \implies (1)$ follows directly by definition. If condition $(b)$ was true, then the problem is not $\delta$-safe since $\overline{I}$ locally progresses into $Q^o$ on $[0, s]$ (implying $\overline I \subseteq \overline Q$) and therefore the domain constraint is trivially satisfied. Similarly, if condition $(c)$ was true, then the condition $\max_{t \in [s, s_1]} (d(R_p^t(z), Q) - \delta(t)) < 0$ ensures that the domain constraint is satisfied, and $d(R^{s_1}_p(z), S^c) < \delta(s_1)$ implies that $\delta$-safety is violated at $(z, s_1) \in \overline{I} \times (0, T]$. 
    
    It remains to prove $(2)$ by assuming that the problem is not $\delta$-safe. If $I \cap Q \cap S^c \neq \emptyset$ then this clearly implies $(a)$, so we may further assume that $I \cap Q \subseteq S$. Hence, there exists $(w, \tau) \in \overline{I} \times (0, T]$ such that $d(R^\tau_{p, B(Q, \delta)}(w), S^c) < \delta(\tau)$. As $\delta(t)$ is finite, this also implies $\tau < t^+(w)$. If $\tau \leq s$, then condition $(b)$ holds and we are done. So suppose that $\tau > s$, the failure of $\delta$-safety at $(w, \tau)$ implies
    \begin{align*}
        &\max_{t \in [s, \tau]} (d(R_p^t(w), Q) - \delta(t)) < 0\\
        &d(R^{\tau}_p(w), S^c) < \delta(\tau)
    \end{align*}
    The claim then follows by continuity and density of $\tilde{\Q} \cap \overline{I}$ in $\overline{I}$ by Lemma \ref{lem: real algebraic numbers are dense}. 
\end{proof}

One can also define a similar notion of $\delta$-robustness for constrained safety problems.

\begin{definition}[$\delta$-Robust Constrained Safety Problem]
    Let $(p, I, S, Q, T)$ be a safety problem and $\delta(t)$ a computable perturbation. The safety problem is $\delta$-robust if exactly one of the following holds:
    \begin{itemize}
        \item $(p, I, S, Q, T)$ is not robustly safe. 
        \item $(p, I, S, Q, T)$ is $\delta$-safe. 
    \end{itemize}
\end{definition}

Similar to the unconstrained case, it can be shown that robust safety is exactly decidable for $\delta$-robust safety problems.

\begin{corollary}[Decidability of $\delta$-Robust Safety]
    There exists a computable algorithm such that given a $\delta$-robust safety problem $(p, I, S, Q, T)$ and a computable $\delta$-perturbation such that $\overline{I}$ locally progresses into $Q^o$, correctly decides the robust safety of the problem.
\end{corollary}

\begin{proof}
    It suffices to run the algorithm given by Theorem \ref{thm: approx decidability of constrained safety}, if the output is robustly safe, then the safety problem is truly robustly safe. Otherwise, the problem is not $\delta$-safe and therefore not robustly safe as it is $\delta$-robust. 
\end{proof}

\section{Conclusion}
By utilizing both numerical and symbolic properties of robust ODEs, this article establishes a complete axiomatization for inductive robust safety, as well as generalizations to constrained safety problems and various approximate decidability results. Importantly, such robust safety properties are exact for initial conditions and do \emph{not} require a positive separation to the unsafe set, allowing for the possibility of inductive proofs of safety for hybrid dynamical systems, leveraging robustness beyond finite time horizons. In particular, this article identifies topological robustness as a suitable notion of inductive robust safety that can be axiomatized, in addition to being a natural extension of state/dynamics robustness, a result of independent interest. 

The results established in this article crucially relied on the central notion of \emph{local progression for sets}, leveraging results from subanalytic geometry to prove a quantitative lower bound for these progressions and thereby obtaining a complete axiomatization for compact-open progression. However, this notion is still relatively unexplored compared to the local progression of points, and it would be interesting to see if there are further applications.

For future work, it would be interesting to study different types of polynomial dynamics and design more efficient proof procedures, as well as sharper quantitative bounds on local progressions. It would also be interesting to determine the decidability of local progression for semialgebraic sets in general.

\bibliographystyle{ACM-Reference-Format}
\bibliography{references}

\newpage
\section*{Appendix}
\appendix

\section{$\dL$ Axiomatization}
\label{sec: dL axiomatization}
The following sound axioms of $\dL$ are used in this article. 
\begin{theorem}[\cite{DBLP:conf/lics/Platzer12b, DBLP:journals/jacm/PlatzerT20, DBLP:journals/fac/TanP21}]
    \label{thm: base axiomatization of dL}
    The following are sound axioms of $\dL$. In axioms \irref{cont}, \irref{dadj}, \irref{bdg}, the variable $y$ is fresh. In axiom \irref{bdg}, we further require that $Q(x)$ is a formula of real arithmetic. 

    \begin{calculus}
        \cinferenceRule[qear|\usebox{\Rval}]{quantifier elimination real arithmetic}
        {\linferenceRule[sequent]
          {}
          {\lsequent[g]{\Gamma}{\Delta}}
        }{$\text{if}~\landfold_{\ausfml\in\Gamma} \ausfml \limply \lorfold_{\busfml\in\Delta} \busfml ~\text{is valid in \LOS[\reals]}$}%
    
        \cinferenceRule[diamond|$\didia{\cdot}$]{diamond axiom}
        {\linferenceRule[equiv]
          {\lnot\dbox{\ausprg}{\lnot \ausfml}}
          {\ddiamond{\ausprg}{\ausfml}}
        }
        {}

        \cinferenceRule[evolved|$\didia{'}$]{evolve}
        {\linferenceRule[equiv]
          {\lexists{t{\geq}0}{\ddiamond{\pupdate{\pumod{x}{y(t)}}}{p(x)}}\hspace{1cm}}
          {\ddiamond{\pevolve{\D{x}=\genDE{x}}}{p(x)}}
        }{$\m{\D{y}(t)=\genDE{y}}, y(0) = x_0$}%

        \cinferenceRule[B|B$'$]{}
        {\linferenceRule[equiv]
          {\lexists{y}{\ddiamond{\pevolvein{x'=f(x)}{Q(x)}}{\rfvar(x,y)}}}
          {\ddiamond{\pevolvein{x'=f(x)}{Q(x)}}{\exists{y}\rfvar(x,y)}}
        }{\text{$y \not\in x$}}
        
        \cinferenceRule[K|K]{K axiom / modal modus ponens} %
        {\linferenceRule[impl]
          {\dbox{\alpha}{(\fvarA \limply \fvarB)}}
          {(\dbox{\alpha}{\fvarA}\limply\dbox{\alpha}{\fvarB})}
        }{}
        \cinferenceRule[V|V]{vacuous $\dbox{}{}$}
         {\linferenceRule[impl]
           {\fvarA}
           {\dbox{\alpha}{\fvarA}}
         }{\text{no free variable of $\fvarA$ is bound by $\alpha$}}
        \cinferenceRule[G|G]{$\dbox{}{}$ generalization} %
        {\linferenceRule[formula]
          {\lsequent{}{\fvarA}}
          {\lsequent{\Gamma}{\dbox{\alpha}{\fvarA}}}
        }{}
        
        \cinferenceRule[dW|dW]{}
        {\linferenceRule
          {\lsequent{\ivr}{P}}
          {\lsequent{\Gamma}{\dbox{\pevolvein{\D{x}=\genDE{x}}{\ivr}}{P}}}
        }{}
        
        \cinferenceRule[dC|dC]{differential cut}%
        {\linferenceRule[sequent]
          {\lsequent[L]{}{\dbox{\pevolvein{\D{x}=\genDE{x}}{\ivr}}{\cusfml}}
          &\lsequent[L]{}{\dbox{\pevolvein{\D{x}=\genDE{x}}{(\ivr\land \cusfml)}}{\ousfml[x]}}}
          {\lsequent[L]{}{\dbox{\pevolvein{\D{x}=\genDE{x}}{\ivr}}{\ousfml[x]}}}
        }{}
        \cinferenceRule[DI|DI]{differential induction}
{\linferenceRule[lpmi]
  {\big(\dbox{\pevolvein{\D{x}=\genDE{x}}{\ivr}}{\ousfml[x]}
  \lbisubjunct (\ivr \limply \ousfml[x])\big)}
  {(\ivr\limply\dbox{\pevolvein{\D{x}=\genDE{x}}{\ivr}}{\der{\ousfml[x]}})}
}
{}%
        \cinferenceRule[DG|DG]{differential ghost variables}
        {\linferenceRule[viuqe]
          {\dbox{\pevolvein{x'=\genDE{x}}{\ivr}}{\ousfml[x]}}
          {\lexists{y}{\dbox{\pevolvein{x'=\genDE{x}\syssep y'=a(x)\cdot y+b(x)}{\ivr}}{\ousfml[x]}}}
        }
        {}
        
        \cinferenceRule[thereAndBack_orig|{[$\&$]}]{there and back quantification}
       {
       \linferenceRule[equiv]
       {\forall t_0 = c_0 \dbox{x' = \theta}{\left(\dbox{x' = -\theta}{\left(c_0 \geq t_0 \rightarrow \chi\right) \rightarrow \phi}\right)}}
        {\dbox{x' = \theta \& \chi}{\phi}}
       }
       {$c_0 \in x$}

        \cinferenceRule[dx|DX]{differential skip}
        {
            \linferenceRule[equiv]
            {\left(Q \rightarrow P \land \dbox{x' = f(x) \& Q}{P}\right)}
            {\dbox{x'= f(x) \& Q}{P}}
        }
        {$x' \notin P, Q$}
        
        \cinferenceRule[uniq|Uniq]{vanilla uniqueness}
        {\linferenceRule[equiv]
        {\left(\ddiamond{x'= f(x) \& Q_1}{P}\right) \land \left(\ddiamond{x'= f(x) \& Q_2}{P}\right)}
        {\ddiamond{x'= f(x) \& Q_1 \land Q_2}{P}}
        }{}

        \cinferenceRule[cont|Cont]{continuity of ODE}
        {\linferenceRule[impl]
        {x = y}
        {\left(\ddiamond{x'= f(x) \& e > 0}{x \neq y} \leftrightarrow e > 0\right)}
        }{$f(x) \neq 0$}

        \cinferenceRule[dadj|Dadj]{reverse flow of ODEs}
        {\linferenceRule[equiv]
        {\ddiamond{y' = -f(y) \& Q(y)}{y = x}}
        {\ddiamond{x' = f(x) \& Q(x)}{x = y}}
        }{}

        \cinferenceRule[ri|RI]{real induction axiom}
        {\linferenceRule[equiv]
        {\forall y\dbox{x' = f(x) \& P \lor x = y}{\left(x = y \rightarrow P \land \ddiamond{x' = f(x) \& P \lor x = y}{x \neq y}\right)}}
        {\dbox{x' = f(x)}{P}}
        }{}
        
        \cinferenceRule[bdg|BDG]{bounded differential ghost}
       {
       \linferenceRule[impll]
       {\dbox{x' = f(x), y' = g(x, y) \& Q(x)}{\norm{y}^2 \leq p(x)}}
        {\left(\dbox{x' = f(x) \& Q(x)}{P(x)} \leftrightarrow \dbox{x' = f(x), y' = g(x, y) \& Q(x)}{P(x)}\right)}
       }
       {}
       
    \end{calculus}
\end{theorem}

The following derivable axioms will also be used.

\begin{theorem}[\cite{DBLP:journals/jacm/PlatzerT20, DBLP:journals/fac/TanP21, DBLP:conf/lics/Platzer12b, DBLP:journals/jacm/PlatzerQ25}]
    \label{thm: bounded differential ghost + differential refinement}
    The following axioms are derivable in $\dL$, where $Q$ is a formula of real arithmetic and $e$ is a term.

    \begin{calculus}
        \cinferenceRule[drd|DR$\didia{\cdot}$]{differential refinement}
        {
        \linferenceRule[impl]
        {\dbox{x' = f(x) \& R}{Q}}
        {\left(\ddiamond{x' = f(x) \& R}{P} \rightarrow \ddiamond{x' = f(x) \& Q}{P}\right)}
        }{}

        \cinferenceRule[drw|dRW$\didia{\cdot}$]{differential refinement}
        {
        \linferenceRule[sequent]
        {\lsequent{R}{Q} & \lsequent{\Gamma}{\ddiamond{x' = f(x) \& R}{P}}}
        {\lsequent{\Gamma}{\ddiamond{x' = f(x) \& Q}{P}}}
        }{}
        
        \cinferenceRule[bdgd|{BDG$\langle\cdot\rangle$}]{bounded differential ghost for diamond}
        {
        \linferenceRule[impll]
        {\dbox{x' = f(x), y' = g(x, y) \& Q(x)}{\norm{y}^2 \leq p(x)}}
        {\left(\ddiamond{x' = f(x) \& Q(x)}{P(x)} \rightarrow \ddiamond{x' = f(x), y' = g(x, y) \& Q(x)}{P(x)}\right)}
        }{}
        
        \cinferenceRule[Kd|$K\didia{\cdot}$]{Diamond Kripke axiom}
        {\linferenceRule[impl]
        {\dbox{\alpha}{\left(\phi \rightarrow \psi\right)}}
        {\left(\ddiamond{\alpha}{\phi} \rightarrow \ddiamond{\alpha}{\psi}\right)}
        }
        {}

        \cinferenceRule[dor|$\didia{}\lor$]{diamond or axiom}
        {\linferenceRule[equiv]
        {\ddiamond{\alpha}{\phi} \lor \ddiamond{\alpha}{\psi}}
        {\ddiamond{\alpha}{\left(\phi \lor \psi\right)}}
        }
        {}

        \cinferenceRule[band|${[]\land}$]{$\dbox{\cdot}{\land}$}
        {\linferenceRule[equiv]
          {\dbox{\alpha}{\phi} \land \dbox{\alpha}{\psi}}
          {\dbox{\alpha}{(\phi \land \psi)}}
        }{}%
        
        \cinferenceRule[enclosure|Enc]{topological enclosure}
        {
        \linferenceRule[sequent]
        {\lsequent{\Gamma}{e \geq 0} & \lsequent{\Gamma}{\dbox{x' = f(x) \& Q \land e \geq 0}{e > 0}}}
        {\lsequent{\Gamma}{\dbox{x' = f(x) \& Q}{e > 0}}}
        }
        {}

        \cinferenceRule[ivt|IVT]{intermediate value theorem}
       {
       \linferenceRule[impll2]
       {e \leq 0 \land \ddiamond{x' = f(x), t' = 1 \& Q}{\left(t = \tau \land e > 0\right)}}
       {\ddiamond{x' = f(x), t' = 1 \& Q \land t < \tau \land e \leq 0}{e = 0}}
       }
       {}
    \end{calculus}
\end{theorem}
The usual FOL proof rules are listed below for completeness \cite{DBLP:journals/jar/Platzer08}.\\
\quad{}
\begin{calculuscollection}
\begin{calculus}
    \cinferenceRule[notl|$\lnot$\leftrule]{$\lnot$ left}
    {\linferenceRule[sequent]
      {\lsequent[L]{}{\asfml}}
      {\lsequent[L]{\lnot \asfml}{}}
    }{}%
    \cinferenceRule[andl|$\land$\leftrule]{$\land$ left}
    {\linferenceRule[sequent]
      {\lsequent[L]{\asfml , \bsfml}{}}
      {\lsequent[L]{\asfml \land \bsfml}{}}
    }{}%
    \cinferenceRule[orl|$\lor$\leftrule]{$\lor$ left}
    {\linferenceRule[sequent]
      {\lsequent[L]{\asfml}{}
        & \lsequent[L]{\bsfml}{}}
      {\lsequent[L]{\asfml \lor \bsfml}{}}
    }{}%
    \cinferenceRule[notr|$\lnot$\rightrule]{$\lnot$ right}
    {\linferenceRule[sequent]
      {\lsequent[L]{\asfml}{}}
      {\lsequent[L]{}{\lnot \asfml}}
    }{}%
    \cinferenceRule[andr|$\land$\rightrule]{$\land$ right}
    {\linferenceRule[sequent]
      {\lsequent[L]{}{\asfml}
        & \lsequent[L]{}{\bsfml}}
      {\lsequent[L]{}{\asfml \land \bsfml}}
    }{}%
    \cinferenceRule[cut|cut]{cut}
    {\linferenceRule[sequent]
      {\lsequent[L]{}{\cusfml}
      &\lsequent[L]{\cusfml}{}}
      {\lsequent[L]{}{}}
    }{}%
    \cinferenceRule[orr|$\lor$\rightrule]{$\lor$ right}
    {\linferenceRule[sequent]
      {\lsequent[L]{}{\asfml, \bsfml}}
      {\lsequent[L]{}{\asfml \lor \bsfml}}
    }{}%
\end{calculus}
\qquad
\begin{calculus}
    \cinferenceRule[implyl|$\limply$\leftrule]{$\limply$ left}
    {\linferenceRule[sequent]
      {\lsequent[L]{}{\asfml}
        & \lsequent[L]{\bsfml}{}}
      {\lsequent[L]{\asfml \limply \bsfml}{}}
    }{}%
    \cinferenceRule[alll|$\forall$\leftrule]{$\lforall{}{}$ left instantiation}
    {\linferenceRule[sequent]
      {\lsequent[L]{p(\astrm)}{}}
      {\lsequent[L]{\lforall{x}{p(x)}}{}}
        \qquad{}\qquad{}
    }{arbitrary term $\astrm$}%
    \cinferenceRule[existsl|$\exists$\leftrule]{$\lexists{}{}$ left}
    {\linferenceRule[sequent]
      {\lsequent[L]{p(y)} {}}
      {\lsequent[L]{\lexists{x}{p(x)}} {}}
    }{\m{y\not\in\Gamma{,}\Delta{,}\lexists{x}{p(x)}}}%
    \cinferenceRule[implyr|$\limply$\rightrule]{$\limply$ right}
    {\linferenceRule[sequent]
      {\lsequent[L]{\asfml}{\bsfml}}
      {\lsequent[L]{}{\asfml \limply \bsfml}}
    }{}%
    \cinferenceRule[allr|$\forall$\rightrule]{$\lforall{}{}$ right}
    {\linferenceRule[sequent]
      {\lsequent[L]{}{p(y)}}
      {\lsequent[L]{}{\lforall{x}{p(x)}}}
    }{\m{y\not\in\Gamma{,}\Delta{,}\lforall{x}{p(x)}}}%
    \cinferenceRule[existsr|$\exists$\rightrule]{$\lexists{}{}$ right}
    {\linferenceRule[sequent]
      {\lsequent[L]{}{p(\astrm)}}
      {\lsequent[L]{}{\lexists{x}{p(x)}}}
    }{arbitrary term $\astrm$}%

    \cinferenceRule[id|id]{identity}
    {\linferenceRule[sequent]
      {\lclose}
      {\lsequent[L]{\asfml}{\asfml}}
    }{}%
\end{calculus}
\end{calculuscollection}

\section{Provable Taylor Approximations}
\label{sec: provable taylor approximation}

\begin{proof}[Proof of Theorem \ref{thm: provable taylor}]
    To motivate the proof, first consider the difference function $R(t) \coloneqq T_n(x_0, t) - e(\phi(x_0, t))$ (where $\phi$ is the flow of $x' = p(x)$). Then this function satisfies (while evolving in $Q$)
    \[R(0) = R'(0) = \cdots = R^{(n)}(0) = 0, \abs{R^{(n + 1)}(t)} \leq M\]
    Consequently, one can bound $R(t)$ by repeatedly bounding its derivatives. In the following derivations, abbreviate $\alpha \equiv x' = p(x), t' = 1 \& Q$.
    \begin{sequentdeduction}
        \linfer[qear+DI]{
            \linfer[qear+DI]{
                \linfer[qear+DI]{
                    \linfer[qear+DI+id]{
                        \lclose
                    }
                    {\lsequent{x = x_0, t = 0, \dbox{\alpha}{\abs{\lied[n + 1]{\genDE{x}}{e}(x)} \leq M}}{\dbox{\alpha}{\abs{\lied[n + 1]{\genDE{x}}{T_n}(x_0, t) - \lied[n + 1]{\genDE{x}}{e}(x)} \leq M}}}
                }
                {\cdots}
            }
            {\lsequent{x = x_0, t = 0, \dbox{\alpha}{\abs{\lied[n + 1]{\genDE{x}}{e}(x)} \leq M}}{\dbox{\alpha}{\abs{\lied{\genDE{x}}{T_n}(x_0, t) - \lied{\genDE{x}}{e}(x)} \leq \frac{Mt^{n}}{n!}}}}
        }
        {\lsequent{x = x_0, t = 0, \dbox{\alpha}{\abs{\lied[n + 1]{\genDE{x}}{e}(x)} \leq M}}{\dbox{\alpha}{\abs{T_n(x_0, t) - e(x)} \leq \frac{Mt^{n + 1}}{(n + 1)!}}}}
    \end{sequentdeduction}
    where each stage of the derivation corresponds to differentiating the difference function, and finally at the $(n + 1)$-th step the $(n + 1)$-th derivative of $T_n(x_0, t)$ will be identically $0$, therefore the premise closes by assumption. 
\end{proof}
\section{Limit Characterization of Robustness}
This section provides a proof of the limit characterization of robust safety. The following definition defines a natural hierarchy of increasingly more general notions of robust safety by considering differential inclusions, $\safe(p, I, S, T)$ is abbreviated as $\safe(I, S)$ in the following. 

\begin{definition}
    \label{def: level n robustness}
    Let $(p, I, S, T)$ be a safety problem. For each $n \in \N$, the problem is said to be $n$-safe, denoted $\safe_n(I, S)$, if there exists some $\eps \in \Q^+$ such that $\safe(I, S)$ holds for the differential inclusion $\norm{x' - p(x)} < \eps t^n$ on $(0, T]$. That is, for all $y \in \R^n$ such that there exists a corresponding time $t_1 \in [0, T]$ and flow $\phi(t) \in C^1([0, t_1], \R^n)$ with $\phi(0) \in I$, $\phi(t_1) = y$ and $\norm{\phi'(t) - p(\phi(t))} < \eps t^n$ for $t \in (0, t_1]$, $y$ necessarily belongs to $S$. 
\end{definition}

The following theorem provides rigorous justification for $\safer(I, S)$ being a ``limiting characterization''.

\begin{theorem}
    \label{thm: limiting characterization of robust safety}
    Let $(p, I, S, T)$ be a safety problem. Then $\safer(I, S)$ holds if and only if there exists some $n \in \N$ such that $\safe_n(I, S)$ holds. I.e
    \[\safer(I, S) \iff \mexists n\safe_n(I, S)\]
\end{theorem}

\begin{proof}
    We first establish the forward implication, that $\safer(I, S)$ implies $ \safe_n(I, S)$ for some $n \in \N$. To this end, first note that $\safer(I, S)$ implies $\lp_p(\overline{I}, S^o)$, thus Theorem \ref{thm: compact progression characterization} item (3) guarantees the existence of some $s \in (0, T)$ and $k \in \N$ such that $d(\phi(x_0, t), S^c) \geq t^k$ on $ (x_0, t) \in \overline{I} \times [0, s]$ where $\phi$ denotes the flow of $x' = p(x)$, it suffices to show that $\safe_k(I, S)$ holds assuming without loss of generality that $s < 1$. As $\safe(\overline{I}, S^o)$ holds on $[s, T]$ by assumption, $\safe_k(I, S)$ will also hold on $[s, T]$ for the differential inclusion $\norm{x' - p(x)} < \eps t^k$ for all sufficiently small $\eps \in \Q^+$ and it remains to handle the interval $[0, s]$. Notice that if $\psi(t) \in C^1([0, s], \R^n)$ satisfies $\norm{\psi'(t) - p(\psi(t))} < \eps t^k$ with $\psi(0) = x_0 \in I$, then Gr\"onwall's inequality \cite{Gronwall_1919} gives 
    \[\norm{\phi(x_0, t) - \psi(t)} \leq C\eps t^{k + 1} e^{Kt}\]
    for absolute constants $C, K$ independent of $\psi$ ($K$ depends on the local Lipschitzness of $p(x)$ by taking $\eps$ sufficiently small). As $t \in [0, s]$ is bounded, this implies 
     \[\norm{\phi(x_0, t) - \psi(t)} \leq C_\eps t^{k + 1} \]
     for some constant $C_\eps$ that tends to $0$ as $\eps \to 0$. A direct application of the triangle inequality then yields
     \[d(\psi(t), S^c) \geq t^{k} - C_\eps t^{k + 1}\]
     by taking $\eps$ sufficiently small such that $C_\eps \leq 1$ yields $d(\psi(t), S^c) \geq t^k(1 - t)$ and therefore establishes the safety of $\psi(t)$ and consequently $\safe_k(I, S)$. \\
     For the converse implication, suppose that $\safe_n(I, S)$ holds for some $n \in \N, \eps > 0$ and assume for the sake of contradiction that $\safer(I, S)$ does not hold. That is, there exists some $x_0 \in \overline{I}$ and $s \in (0, T]$ such that $\phi(x_0, s) \notin S^o$ ($I \subseteq S$ holds since $\safe_n(I, S)$ holds). First consider the following matrix-valued IVP
     \begin{align*}
         &W_{x_0}'(t) = A_{x_0}(t)W_{x_0}(t)\\
         &W_{x_0}(0) = I_n
     \end{align*}
     where $I_n$ denotes the $n \times n$ identity matrix and $A_{x_0}(t) = \frac{\partial p}{\partial x} \vert_{\phi(x_0, t)}$ is the linearization of $p$ at $x = \phi(x_0, t)$. As this differential equation is linear in $W_{x_0}(t)$ and $A_{x_0}(t)$ is well-defined on $[0, s]$, let $W_{x_0}(t)$ be the unique solution to this IVP. Now consider the following parametrized IVP (implicitly assuming the clock variable $t' = 1$ is included) where $\vec{\eps} \in \R^n$ is some constant perturbation vector. 
     \begin{align}
        \label{eqn: app ode}
         &x' = p(x) + t^nW_{x_0}(t)\vec{\eps}\\
         &x(0) = x_0
     \end{align}
    Notice that for a fixed $\vec{\eps}$, the corresponding solution $x_{\vec{\eps}}(t)$ (which necessarily exists by picking $\norm{\vec{\eps}}$ sufficiently small) of this IVP satisfies the following for $t \in [0, s]$
     \[\norm{x'_{\vec{\eps}} - p(x_{\vec{\eps}})} = \norm{t^nW_{x_0}(t)\vec{\eps}} \leq K_{x_0}\norm{\vec{\eps}}t^n\]
     where $K_{x_0}$ is some constant independent of $\vec{\eps}$ bounding the operator norm of $W_{x_0}(t)$ on $t \in [0, s]$. From now on, we assume that the domain of $\vec{\eps}$ is small enough such that $K_{z_0}\norm{\vec{\eps}} < \eps$ for all $z_0 \in B(x_0, 1)$. The proof proceeds by establishing that the following map is locally invertible around $\vec{\eps} = \vec{0}$\footnote{Readers familiar with basic results from control theory will recognize that this follows from the (local) controllability of $x' = p(x) + t^n\vec{\eps}$.}:
     \[\vec{\eps} \mapsto x_{\vec{\eps}}(s)\]
     To this end, first define the function (for $(z_0, \vec{\eps})$ locally around $(x_0, \vec{0})$ so that the map is well-defined):
     \[F(z_0, t, \vec{\eps}) \coloneqq z_{\vec{\eps}}(t)\]
     where $z_{\vec{\eps}}$ is the solution to (\ref{eqn: app ode}) with $x_0 \coloneqq z_0$. Note that this function is smooth by smooth dependence on initial conditions for ODEs. By the inverse function theorem applied to $\vec{\eps} \mapsto x_{\vec{\eps}}(s)$, it suffices to show that the Jacobian $Z(z_0, t) \coloneqq D_{\vec{\eps}}F(z_0, t, \vec{0})$ is non-singular at $t = s$ for all $z_0$. Differentiating $F$ with respect to both time and $\vec{\eps}$ and evaluating at $\vec{\eps} = \vec{0}$ yields 
     \[Z'(z_0, t) = A_{z_0}(t)Z(z_0, t) + t^nW_{z_0}(t)\]
     Since the mapping $\vec{\eps} \mapsto F(z_0, 0, \vec{\eps})$ is constant, $Z(z_0, t)$ satisfies the initial condition $Z(z_0, 0) = \mathbf{0}$. Thus, the overall solution of $Z(z_0, t)$ is 
     \[Z(z_0, t) = W_{z_0}(t) \int_{0}^t \tau^{n}W_{z_0}^{-1}(\tau)W_{z_0}(\tau) d\tau = \frac{t^{n + 1}}{n + 1}W_{z_0}(t)\]
     which is non-singular at $t = s > 0$, establishing local invertibility of the map $\vec{\eps} \mapsto z_{\vec{\eps}}(s)$. By a quantitative version of the inverse function theorem \cite[Theorem~2.9.4]{hubbard:hal-01297648} that lower-bounds the region of invertibility around $F(x_0, s, \vec{0}) = x_{\vec{0}}(s)$ using the (local) Lipschitz constant of the map \cite{DBLP:journals/tac/Ratschan18}, and the (local) smoothness of $F(z_0, t, \vec{\eps})$, it follows that there exists some $R > 0$ such that the following inclusion holds for all $z_0$ sufficiently close to $x_0$.
     \[B(F(z_0, s, \vec{0}), R) \subseteq F(z_0, s, B(\vec{0}, r))\]
     where $r > 0$ is chosen small enough such that $K_{z_0} r < \eps$. Finally, since $\phi(x_0, s) \notin S^o$ by assumption, it is arbitrarily close to $S^c$. Hence, by picking $z_0 \in I$ sufficiently close to $x_0 \in \overline{I}$ and satisfying $F(x_0, s, \vec{0}) \in B(F(z_0, s, \vec{0}), R)$, we can find some $u_0 \in S^c$ such that $u_0 \in B(F(z_0, s, \vec{0}), R)$, and the local invertibility on $B(F(z_0, s, \vec{0}), R)$ implies the existence of some perturbation $t^nW_{z_0}(t)\vec{\eps}$ for which $z_0$ will be unsafe, contradicting our assumption of $\safe_n(I, S)$.
\end{proof}

\end{document}